\documentclass[sigconf,nonacm]{acmart}
\AtBeginDocument{
  }
\usepackage{mathrsfs}
\usepackage{multirow}
\usepackage{amsthm}
\usepackage{dsfont}

\newtheorem{lemma}{Lemma}
\newtheorem{definition}{Definition}

\usepackage{algorithm}
\usepackage{algorithmic}
\usepackage{makecell}
\usepackage{tabularx}
\usepackage{subcaption}
\usepackage{array}
\usepackage{tcolorbox}
\usepackage{varwidth}
\tcbuselibrary{breakable}
\tcbuselibrary{skins}
\NewTColorBox{theobox}{o m +o}{
    enhanced,
    colframe=gray!0!white,
    interior empty,
    coltitle=white,
    fonttitle=\mdseries,
    colbacktitle=gray,
    extras broken={frame empty,interior empty},
    borderline={0.5mm}{0mm}{gray},
    sharp corners=downhill,
    top=4mm,
    before skip=3.5mm,
    attach boxed title to top left={yshift=-3mm,xshift=5mm},
    boxed title style={boxrule=0pt,sharp corners=all},varwidth boxed title,
    IfNoValueTF={#1}{title=#2}{title=#2:~#1},
    IfNoValueTF={#3}{}{#3}
}

\begin{document}

\title{Provably Secure Steganography Based on List Decoding}

\author{Kaiyi Pang}
\affiliation{
  \institution{Tsinghua University}
  \city{Beijing}
  \country{China}}
\email{pky@mail.tsinghua.edu.cn}

\author{Minhao Bai}
\affiliation{
  \institution{Tsinghua University}
  \city{Beijing}
  \country{China}}
\email{bmh22@mails.tsinghua.edu.cn}

\begin{abstract}
Steganography embeds secret messages in seemingly innocuous carriers for covert communication under surveillance. Current Provably Secure Steganography (PSS) schemes based on language models can guarantee computational indistinguishability between the covertext and stegotext.
However, achieving high embedding capacity remains a challenge for existing PSS. The inefficient entropy utilization renders them not well-suited for  Large Language Models (LLMs), whose inherent low-entropy tendencies severely constrain feasible embedding capacity.
To address this, we propose a provably secure steganography scheme with a theoretically proved high capacity. Our scheme is based on the concept of list decoding: it maintains a set of candidates that contain the correct secret message, instead of directly finding the correct message with more effort. This strategy fully utilizes the information content of the generated text, yielding higher capacity.
To ensure the correctness of our scheme, we further introduce a suffix-matching mechanism to distinguish the correct secret message from the candidates.
We provide theoretical proofs for both the security and correctness of our scheme, alongside a derivation of its theoretical capacity lower bound.
 Our approach is plug-and-play, requiring only a direct replacement of the model's standard random sampling module.
Experiments on three LLMs and seven PSS baselines demonstrate that our method achieves computational efficiency comparable to prior PSS schemes while delivering a substantial improvement in embedding capacity.

\end{abstract}

\maketitle

\section{Introduction}

Steganography is a technique for embedding and transmitting private messages within seemingly innocuous carriers such as text \cite{yang2018rnn, meteor2021, ding2023discop}. Its primary goal is to conceal not only the content of the secret messages but also the very fact that secret communication is occurring. Provably Secure Steganography (PSS)~\cite{hopper2004toward,meteor2021} represents the state-of-the-art paradigm, ensuring that model-generated stegotext is computationally indistinguishable from covertext sampled from the underlying model distribution.

 As a covert communication mechanism, steganography is required not only to satisfy security guarantees but also to maximize communication efficiency. In the context of provably secure generative text steganography~\cite{meteor2021,ding2023discop,shimmer}, communication efficiency is typically defined as the utilization rate of the entropy of the model distribution. Intuitively, a steganographic scheme with higher capacity can embed more secret messages into shorter stegotext, there-by reducing communication frequency. While non-secure methods~\cite{ziegler2019neural,yang2020vae,yang2018rnn,fang2017generating} often boost capacity by distorting the generation distribution, this comes at the unacceptable cost of compromising imperceptibility.

In contrast, existing provable secure steganography methods rely primarily on algorithm constructions to improve capacity but often lack principled optimization strategies. In several representative schemes~\cite{ding2023discop,liao2025framework,meteor2021}, the payload associated with a generation step is recovered using only the current token and the current-step distribution. This creates a structural bottleneck. When multiple message candidates map to the same token, the encoder must either embed only a short common prefix or forgo embedding at that step altogether. Either choice leaves part of the token’s information content unused, reducing entropy utilization. This limitation is especially problematic for Large Language Models (LLMs), whose next-token distributions are often highly peaked, so many tokens carry little entropy and inherently constrain the embedding capacity.

Recent work partially alleviates this issue. SparSamp~\cite{wang2025sparsamp} and Shimmer~\cite{shimmer} allow decoding evidence to accumulate across multiple tokens and achieve noticeable capacity gains. However, these methods remain limited:
SparSamp lacks a theoretical capacity analysis and stops embedding actively when ambiguity persists, while Shimmer loses information during interval splitting and lacks a mechanism to maintain multiple decoding paths without incurring exponential blowup. More broadly, interval-based representations often require repeated conversions between binary strings and fractional values~\cite{wang2025sparsamp,shimmer}, which can be cumbersome and numerically delicate in practice.
Overall, a fundamental gap remains: we need a principled mechanism that continuously exploits entropy by maintaining multiple candidate states within bounded resources, rather than discarding them. Moreover, interval-based representations require frequent conversions between binary strings and fractional values~\cite{wang2025sparsamp,shimmer}, which can be cumbersome in practice.

Inspired by list decoding~\cite{listdecoding1,listdecoding2}  in information theory, we propose a principled steganography scheme that more fully exploits the entropy inherent in the model’s distribution. The basic idea of list decoding is that, when a unique secret message cannot be immediately identified, the decoder instead outputs a compact list of candidates guaranteed to contain the true message. A classic example is the Guruswami–Sudan algorithm~\cite{listdecoding1}, which recovers a small list containing the correct codeword even under high error rates. Adapting this paradigm to steganography, our scheme maintains a ``list of candidate messages'' throughout the generation process. This allows us to defer decoding decisions, thereby avoiding the capacity loss inherent in prior constructions~\cite{meteor2021,ding2023discop,liao2025framework} that prioritize immediate unique decoding.

In this paper, we introduce a provably secure steganographic scheme grounded in a list-decoding perspective. By strictly following the model's sampling distribution to expand and filter the candidate list, our method allows the embedding rate to approach the theoretical limit dictated by the entropy of the generated text. Furthermore, to resolve the residual ambiguity inherent in list decoding, we propose a suffix-matching strategy. This mechanism enables the receiver to efficiently identify the unique secret message from the candidate list, ensuring security and correctness while maintaining an embedding capacity near the model's entropy.

The main contributions of this paper are as follows:
\begin{itemize}
\item We propose a provably secure steganographic scheme based on list decoding, and provide a theoretical analysis showing that its embedding capacity can approach the entropy limit.
\item We introduce a suffix-matching mechanism to assist in fast and correct decoding, and we provide formal proofs of correctness.
\item Extensive experiments demonstrate that our method achieves higher capacity than state-of-the-art PSS schemes while preserving comparable computational efficiency and text quality.

\end{itemize}
\section{Related Work}
\label{relatedworks}

We briefly review representative provably secure steganographic schemes for language models, focusing on how each algorithm embeds and recovers payload information and on the main bottlenecks that limit capacity.

\subsubsection{METEOR}
METEOR \cite{meteor2021}, proposed by Kaptchuk et al., is the first practical provably secure steganographic scheme for language models. It encrypts the secret bitstream using a key-generated random mask via bitwise XOR, interprets the result as a real number, and embeds it by sampling tokens according to the cumulative distribution function of the language model. The masking process is repeated independently at each generation step, and the common prefix of tokens falling within the range represents the secret information that can be embedded at this time step. Although computationally secure, METEOR does not fully exploit the model’s entropy. Re-randomizing the secret bits at every step truncates the coding process and discards residual entropy, which significantly limits embedding capacity—especially for low-entropy token distributions common in LLMs. To address this issue, a reordering algorithm, METEOR(R.), was later proposed to improve the expected embedding capacity.

\subsubsection{DISCOP}

Ding et al. proposed a provably secure steganography method based on multiple shifted copies of the model distribution~\cite{ding2023discop}. When distinct random offsets map to different tokens, multiple secret bits can be embedded in a single step, and the decoder can recover them directly from the generated token. However, substantial overlap among distribution copies severely limits the embedding rate, which is asymptotically bounded by the minimum entropy. To alleviate this limitation, the authors further proposed a Huffman-tree-based variant DISCOP (R.) to improve capacity.

\subsubsection{FDPSS}
Liao et al. proposed a provably secure steganographic design framework, FDPSS~\cite{liao2025framework}, and as an instantiation of this framework, developed a more efficient, capacity-oriented encoding and decoding algorithm. We call this differential-based algorithm GROUP. The encoder first applies a differencing operation to the model’s output distribution and then represents it as a mixture of uniform distributions. During encoding, a random number is sampled to determine which uniform component to draw from, and the secret message is used to select a specific token within the chosen component.
 This method is efficient because both encoding and decoding depend only on the distribution at the current time step.
\subsubsection{Shimmer}
To address the inefficiency of METEOR~\cite{meteor2021}, Bai et al. proposed Shimmer~\cite{shimmer}, a provably secure steganographic scheme that incorporates an entropy-collection mechanism. Its security relies on the indistinguishability between $r \sim \mathsf{Unif}[0,1)$ and $r + B \bmod 1 \sim \mathsf{Unif}[0,1)$, where $B$ denotes the fractional representation of the secret message. During encoding, the sampler adds $B$ to the random variable $r$ and applies inverse transform sampling to generate tokens; the interval shift enables recovery of the secret interval. By progressively merging intervals across steps, Shimmer embeds the secret message while mitigating interval splitting through additional mechanisms. Overall, Shimmer achieves high entropy utilization with computational security and demonstrates how carrying interval state across multiple steps can improve capacity.

\subsubsection{SparSamp}
To address both the capacity and efficiency limitations of DISCOP~\cite{ding2023discop}, Wang et al. proposed SparSamp~\cite{wang2025sparsamp}, a provably secure steganographic scheme that improves decoding accuracy by inserting gaps between sampling intervals without additional time complexity. SparSamp encodes a fixed-length $L$-bit secret by generating $2^L$ samples and repeatedly filtering candidate messages across multiple steps until only the true secret remains. While SparSamp achieves state-of-the-art embedding capacity, it lacks a formal capacity proof, and its fixed-length design still underutilizes joint entropy, leaving room for further improvement.

\section{Method}

\subsection{Prelimilary}

Following the notion of Hopper~\cite{hopper2004toward} and Kaptchuk~\cite{meteor2021}, we formally define steganography scheme as follows:
\begin{definition}[Steganography scheme]
A steganography scheme $\Sigma_{\mathsf{D}}$ on a covertext distribution $\mathsf{D}$ is a triple of (probabilistic) algorithms $(\mathsf{KeyGen},  \mathsf{Enc}_{sk,\mathsf{D}}, \mathsf{Dec}_{sk,\mathsf{D}})$.
    \begin{itemize}
    \item $\mathsf{KeyGen}(1^\lambda)$ is a randomized algorithm that takes an arbitrary input of length $\lambda$ and generates the key $sk$ shared between the sender and receiver.
    \item $\mathsf{Enc}_{sk,\mathsf{D}}(m,h)$ is a keyed randomized algorithm that takes as input a secret message $m \in \{0,1\}^{*}$ and the history $h$ and outputs a stegotext $t$.
    \item $\mathsf{Dec}_{sk,\mathsf{D}}(t,h)$ is a keyed randomized algorithm that takes as input the stegotext $t$ and the history $h$ and outputs the secret message $m$.
\end{itemize}
\end{definition}

In this paper, we mainly use the language models to obtain the covertext distribution $\mathsf{D}$. So the encoding and decoding algorithms can be written as $\mathsf{Enc}_{sk,\mathsf{Model}}(m,h)$ and $\mathsf{Dec}_{sk,\mathsf{Model}}(t,h)$
Basically, a useful steganography scheme $\Sigma_{\mathsf{Model}}$ should satisfy \textit{correctness} and \textit{security}.

\begin{definition}[Correctness]
   A steganography scheme $\sum_{\mathsf{Model}}$ is \textit{correct} if for any history $h$ and any message $m$, the stegotext generated by $\mathsf{Enc}_{sk,\mathsf{Model}}(m,h)$ should be correctly decoded by the decoding algorithm only with negligible probability of error.
    \begin{equation}
\begin{aligned}
    \Pr[\mathsf{Dec}_{sk,\mathsf{Model}}(\mathsf{Enc}_{sk,\mathsf{Model}}(m,h),h)= m] & \\
    \geq 1-\mathsf{negl}(\lambda), &
\end{aligned}
\end{equation}
where $\mathsf{negl}(\lambda)$ is the negligible function that correlates to the secure parameter $\lambda$.
\end{definition}

\paragraph{Threat model}
We consider a probabilistic polynomial-time (PPT) warden
(adversary) who monitors the communication channel and attempts to
determine whether an observed message is a stegotext or an innocent
covertext. Following the standard chosen-hiddentext attack (CHA)
setting~\cite{hopper2004toward,meteor2021}, the warden knows the full
public description of the scheme, including all public hyperparameters
that determine the covertext channel (e.g., the language model and
tokenizer). The only secret information withheld from the warden is the
shared key~$sk$. The warden is given oracle access either to the stego
encoder $\mathsf{Enc}_{sk,\mathsf{Model}}(\cdot,\cdot)$ or to the
model sampler $\mathsf{Model}(\cdot)$, and try to distinguish between the
two.

\paragraph{Security definition}
We define the security of steganographic system against chosen hiddentext attacks~\cite{hopper2004toward,meteor2021}. Intuitively, $\Sigma_{\mathsf{Model}}$ is secure if no probabilistic polynomial-time (PPT) adversary can distinguish, with non-negligible advantage, between access to $\mathsf{Enc}_{sk,\mathsf{Model}}(\cdot,\cdot)$ and access to the model sampler $\mathsf{Model}(\cdot)$, given outputs of the same length. More formally,
\begin{definition}[Security]
A steganography scheme $\Sigma_{\mathsf{Model}}$ is secure if for any probabilistic polynomial-time adversary $\mathcal{A}$, he cannot effectively distinguish the stegotext generated by $\mathsf{Enc}$ and the covertext generated by $\mathsf{Model}$.
    \begin{equation}
        \begin{aligned}
            |\Pr[\mathcal{A}^{\mathsf{Enc}_{sk,\mathsf{Model}}(\cdot,\cdot)}(1^\lambda)=1]-\Pr[\mathcal{A}^{\mathsf{Model}(\cdot)}(1^\lambda)=1]| & \\
            \leq \mathsf{negl}(\lambda). &
        \end{aligned}
    \end{equation}
    Here the $1^\lambda$ represents the internal randomness of adversary $\mathcal{A}$, and $\mathcal{A}$ outputs $1$ if he recognizes the output of $\mathsf{Enc}$ or $\mathsf{Model}$ as stegotext.
\end{definition}

\paragraph{Access assumptions}
Our construction assumes a symmetric setting in which both the sender and the receiver can access the full conditional next-token distribution $D=\mathsf{Model}(h)$ for any history $h$. Such access is required to perform alias sampling and to synchronize the message-to-token mapping between the two parties. Pure black-box settings and other asymmetric scenarios(such as the sender and receiver have the different ) are beyond the scope of this work.

\begin{figure*}[htb]
	\centering
	\includegraphics[width=\linewidth]{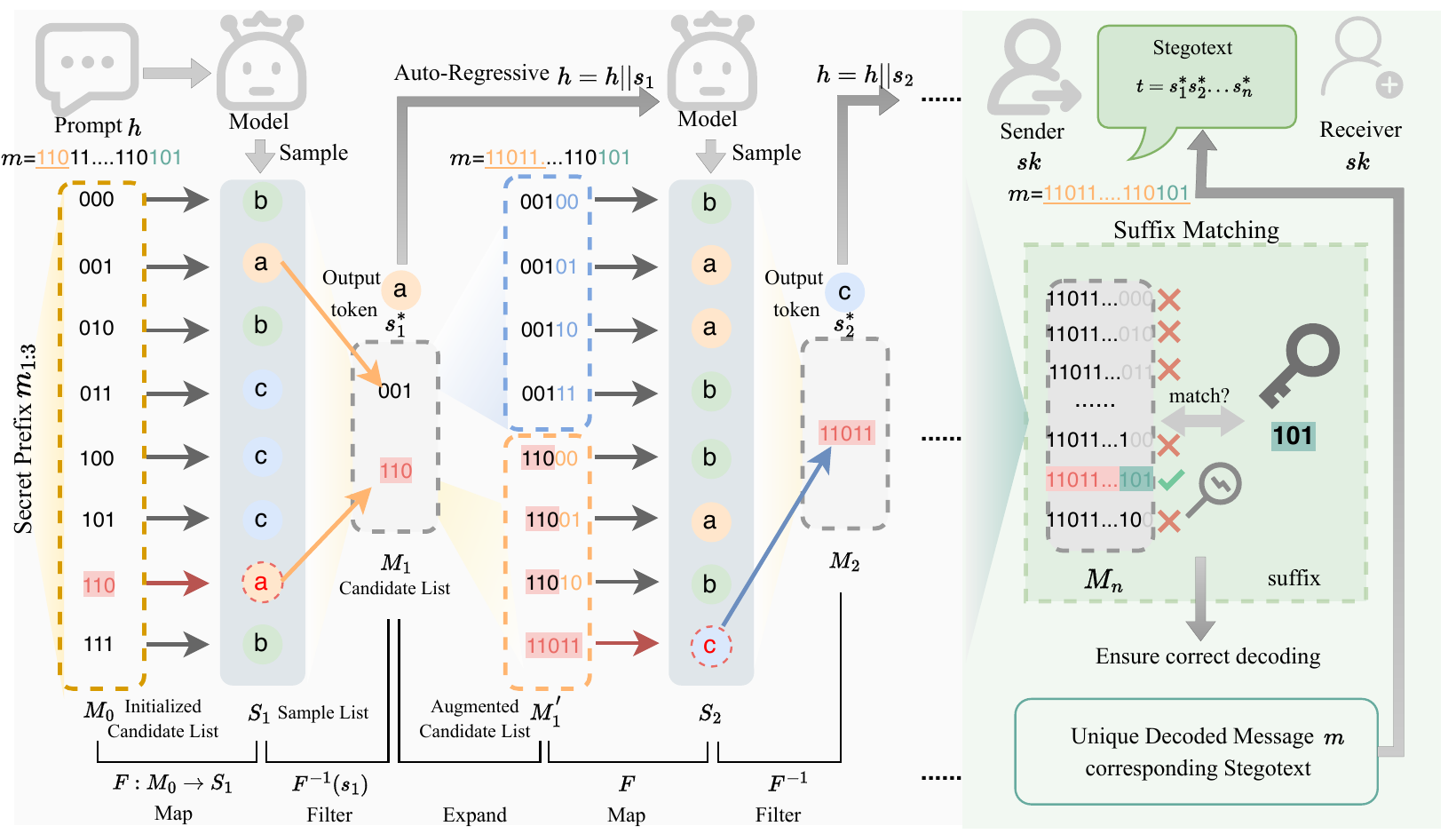}
	\caption{Toy illustration to our embedding process. The prefix of the secret message $m$ defines an initial candidate list $M_0$. After mapping through the language model to obtain $S$, candidates whose prefixes correspond to the same tokens in $S$ are retained, forming the filtered set $M_1$, and the token $s_1$ associated with the secret message is selected for output. Autoregressive generation proceeds in this manner until both the secret message and its suffix are fully embedded. Finally, suffix matching is applied to the filtered candidate list $M_t$, enabling rapid elimination of non-secret candidates and guaranteeing unique decodability of the secret message from the stegotext.}
	\label{example}
\end{figure*}

\paragraph{Useful inequalities}
The following two classical inequalities are used frequently in our
analysis.
\begin{lemma}[{\bf Hoeffding's inequality}]\label{hoeffding-inequality}
$X_1, X_2, \cdot\cdot\cdot , X_n$ are independent and identical random variables, and each $X_i$ is bounded by $[l,h]$.
Let $X = \frac{1}{n}\sum_{i=1}^n X_i$ and $\mu = \mathbb{E}[X]$, the probability that the sample mean $X$ deviates from the theoretical mean $\mathbb{E}[X]$ up to $t$ is
\begin{align}
    \mathbb{P}\left(\left|X - \mu\right| \geq t\right) \leq 2\exp\left(-\frac{2nt^2}{\left(h-l\right)^2}\right).
\end{align}
\end{lemma}

\begin{lemma}[{\bf Cauchy-Schwarz inequality}]\label{cauchy-schwarz-inequality}
For any real sequences $\{a_i\}_{i=1}^l$ and $\{b_i\}_{i=1}^l$, the square of the sum of their pairwise products is bounded by the product of the sums of their squares:
\begin{align}
    \left( \sum_{i=1}^l a_i b_i \right)^2 \leq \left( \sum_{i=1}^l a_i^2 \right) \left( \sum_{i=1}^l b_i^2 \right).
\end{align}
In the special case where $a_i = 1$ for all $i$, the inequality simplifies to:
\begin{align}
    \left( \sum_{i=1}^l b_i \right)^2 \leq l \sum_{i=1}^l b_i^2.
\end{align}
\end{lemma}

\subsection{Our Steganography Scheme}

\subsubsection{Intuition}

As discussed in Section~\ref{relatedworks}, many existing provably secure steganography methods recover payload information using only the current token and the current-step distribution. This limitation often forces the sender to skip embedding at certain steps to avoid ambiguity, resulting in suboptimal capacity usage.
To overcome this, we look to list decoding~\cite{listdecoding1,listdecoding2} for inspiration. Fundamentally, list decoding relaxes the constraint of outputting a single message, allowing the algorithm to produce a list of possibilities, one of which is the correct message.
Inspired by list decoding, our proposed scheme does not require the receiver to uniquely decode the secret at the immediate token level. Instead, we allow the secret information to be embedded continuously at every time step, even if immediate extraction is ambiguous, aiming to maximize the utilization of the generative distribution's entropy and expand steganographic capacity. The correct secret message persists within a set of surviving candidates across time steps, while incorrect candidates (those conflicting with the secret) are naturally "filtered" out during the sequential token selection process. To ensure correctness, we introduce a suffix matching mechanism. The combination of this natural filtering and suffix matching theoretically guarantees that the receiver can uniquely recover the message after observing the complete stegotext (theoretical proofs regarding correctness are detailed in Section~\ref{proof_of_correctness}). Figure~\ref{example} also illustrates an example of this encoding process.

\subsubsection{Codec}
The workflow of the encoding and decoding algorithm can be summarized by a four-step process: \textbf{"Mapping, Filtering, Expanding, and Matching"}.

\textbf{Mapping}.
During the sampling process, the sender must generate a number of samples equal to the number of candidate secret messages in the current candidate list. To share these samples between the sender and receiver, a pseudorandom generator is employed to generate all random numbers used during the sampling process. To effectively gather a sufficient quantity of samples for establishing a direct mapping between secret messages and samples, some fast sampling techniques such as the Alias sampling method \cite{alisa1,alisa2}, detailed in Appendix \ref{alias-method}, can be employed. Assuming the list of candidate secret messages at a certain step is $M$, we generate $|M|$ samples $\{s_i\}_{i = 1}^{|M|}$ to construct the sample list $S$ by the codebook mapping $F(m_i) := s_i, m_i \in M$. Subsequently, based on the current prefix of the embedded bitstring (secret bits together with the validation suffix), denoted $m_{1:l}$, the sample $s^* = F(m_{1:l})$ is output. Since the generated samples are mutually independent, randomly selecting and outputting one is equivalent to the model's normal sampling process (detail is provided in the security proof section~\ref{security_proof}).

\textbf{Filtering}.
\label{subsec:list-decoding-filtering}
Transmitting a single sample $s^*$ does not uniquely determine the secret message because multiple different candidate messages may map to the same sample; that is, the pre-image set $F^{-1}(s^*) = \{ m_i : F(m_i) = s^* \}$ may contain multiple elements. We can view the embedding process as a filtering process that \textbf{progressively eliminates candidate messages}.

At the start of embedding, the sender initializes a candidate list $M_0$ containing all possible $N$-bit strings, with a size of $2^N$. Clearly, the $N$-bit prefix of the real secret message $m^*_{1:N}$ is in $M_0$. The sender then constructs the Sample List $S$ by message-to-token mapping $F_1$, and sends the sample $s_1^* = F_1(m^*_{1:N})$. Upon receiving $s_1^*$, the receiver identifies all candidates capable of producing this sample using the inverse mapping, which is the pre-image set of $s^*:M_1 = F_1^{-1}(s^*)$. At this point, the prefix $m^*_{1:N}$ must also belong to $M_1$. Since the pre-image set $M_1$ is a subset of $M_0$, we have $|M_1| \leq |M_0|$, narrowing the candidate list.

Further analysis reveals the probability and extent of this list reduction. Assuming the output token is $s$ with probability $D(s)$, and we draw $|M_0|$ i.i.d.\ samples to construct the mapping, on average, a proportion of $D(s)$ of the candidate messages remains, i.e., $\mathbb{E}\left[|M_1|\right] = D(s) \cdot |M_0|$. According to Hoeffding's inequality, the upper bound of the ratio $\frac{|M_1|}{|M_0|}$ is given by:
\begin{align}\label{eq3.1}
    \Pr\left[\frac{|M_1|}{|M_0|} \geq D(s) + \delta \right] \leq \exp\left(-2\delta ^2|M_0|\right).
\end{align}

Therefore, determining a sufficiently large $M_0$ guarantees the elimination of a $(1-D(s))$ proportion of candidates. For instance, when $|M_0| = 2^{20}$, the deviation $\delta$ rarely exceeds $0.5\%$. In practical applications, the number of samples should be maximized without affecting the execution time, and a very large candidate list should be maintained at each step to ensure filtering efficiency.

From an information-theoretic perspective, although the secret message is not uniquely determined at this stage, the possible choices are reduced. for the decoder, the entropy of the secret message drops from $\log_2(|M_0|)$ to $\log_2(|M_1|)$, a reduction of $\log_2\left(\frac{|M_0|}{|M_1|}\right) \approx \log_2(D(t))$. This approximates the information of the token.

\textbf{Expanding}
After filtering, the size of the candidate list decreases. To embed more secret bits and maintain the candidate list size at a high level, the list should be \textbf{expanded}. A feasible strategy is as follows: for each candidate message $m$ in the current list $M_i$, append bits `0' and `1' to generate two new messages $m \| 0$ and $m \| 1$. These form the augmented list $M_n'$. Clearly, if the $k$-bit prefix $m^*_{1:k}$ belongs to $M_i$, then the $(k+1)$-bit prefix $m^*_{1:k+1}$ belongs to $M_i'$. To prevent memory overflow from exponential growth, an upper limit $2^N$ (e.g., $2^{20}$) is set. If the expanded list size is still less than half of the limit ($2^{N-1}$), the expansion operation repeats until the size exceeds $2^{N-1}$. Consequently, the filtering step is always performed with more than $2^{N-1}$ candidate messages, effectively ensuring that the entropy dissipated in each step remains low. The complete filtering-expanding process can be viewed as follows:
\begin{align}
    M_0 \xrightarrow{\text{Filter}}     \cdots \xrightarrow{\text{Filter}}  M_{i} \xrightarrow{\text{Expand } e_1 \text{ times}} M'_{i} \xrightarrow{\text{Filter}} M_{i + 1} \xrightarrow{\text{Filter}} \cdots \notag
\end{align}

\begin{algorithm}[htb]
    \caption{$\mathsf{Enc}_{sk,\mathsf{Model},2^N}(h,m^*) \xrightarrow[]{} t$}
    \renewcommand{\algorithmicrequire}{\textbf{Input:}}
    \renewcommand{\algorithmicensure}{\textbf{Output:}}
    \label{alias-encode}
    \begin{algorithmic}[1]
    \REQUIRE Secret key $sk$, Language Model $\mathsf{Model}$, History $h$, Secret message $m^*$, Maximum list length $2^N$;
    \ENSURE Stegotext $t$.
    \STATE $M,t \longleftarrow \{0,1\}^N, \emptyset$ \COMMENT{Initialize candidate message list $M$ and stegotext $t$}
    \STATE $l,cnt \longleftarrow N,0$ \COMMENT{Initialize message-length counter $l$ and counter $cnt$}
    \STATE $suf \longleftarrow \mathsf{PRG}_{sk}(\cdot)$ \COMMENT{Shared validation suffix}
    \STATE $m \longleftarrow m^* \| suf$ \COMMENT{Add suffix}
    \WHILE{$l \leq |m|$}
    \STATE $cnt \longleftarrow cnt + 1$
    \STATE $D \longleftarrow \mathsf{Model}(h)$ \COMMENT{Obtain distribution}
    \STATE $\{r_i\}_{i = 1}^{2|M|} \longleftarrow \mathsf{PRG}_{sk}(\cdot)$ \COMMENT{Derive pseudorandomness}
    \STATE $\{s_i\}_{i = 1}^{|M|} \longleftarrow \mathsf{AliSample}_{|M|}\left(D,\{r\}_{i = 1}^{2|M|}\right)$ \COMMENT{Execute Alias sampling to obtain $|M|$ samples}
    \FOR{$i=1$ {\bf to} $|M|$}
        \STATE Let $F(m_i) := s_i$ \COMMENT{Construct a mapping}
    \ENDFOR
    \STATE $s^* \longleftarrow F(m_{1:l})$ \COMMENT{Select token corresponding to current message prefix}
    \STATE $t,h \longleftarrow t \| s^*, h\| s^*$ \COMMENT{Append new token $s^*$ to stegotext and history}

    \STATE $M \longleftarrow F^{-1}(s^*)$ \COMMENT{Find pre-image set of $s^*$ and update candidate message list $M$}
    \WHILE{$|M| \leq 2^{N - 1}$}
        \STATE $M \longleftarrow \left\{m \| bit : m \in M, bit \in \{0,1\}\right\}$ \COMMENT{Expand candidate message list $M$}
        \STATE $l \longleftarrow l + 1$ \COMMENT{Increment candidate message length after expansion}
    \ENDWHILE
    \ENDWHILE
    \RETURN $t$
    \end{algorithmic}
\end{algorithm}

\textbf{Matching}
\label{subsec:checksum-embedding-decoding}
The "Filter-Expand" process repeats until all bits of the secret message are embedded. At this stage, the candidate list $M_n$ still contains multiple messages, one of which is the complete real secret message $m^*$. To ensure the receiver can uniquely determine $m^*$, a {validation phase} is appended. Concretely, the sender and receiver derive a pseudorandom validation suffix $suf$ from the shared key, append it to the payload, and continue embedding. Since the suffix is pseudorandom and sufficiently long (detailed analysis in Section~\ref{proof_of_correctness}), with high probability only the candidate corresponding to the real payload will fully match it, enabling unique decoding.

The decoding algorithm is essentially the symmetric inverse of the encoding process. Detailed procedures for the encoding algorithm, $\mathsf{Enc}_{sk,\mathsf{Model},2^N}(h,m^*)$, and the decoding algorithm,\\ $\mathsf{Dec}_{sk,\mathsf{Model},2^N}(h,t,|m^*|)$, are presented in Algorithm \ref{alias-encode} and Algorithm \ref{alias-decode}, respectively. To account for potential hardware constraints, we denote the maximum candidate list length, which also serves as the maximum sample count, as the subscript parameter $2^N$.

\begin{algorithm}[ht]
    \caption{$\mathsf{Dec}_{sk,\mathsf{Model},2^N}(h,t,|m^*|) \xrightarrow[]{} m^*$}
    \renewcommand{\algorithmicrequire}{\textbf{Input:}}
    \renewcommand{\algorithmicensure}{\textbf{Output:}}
    \label{alias-decode}
    \begin{algorithmic}[1]
    \REQUIRE Secret key $sk$, Language Model $\mathsf{Model}$, History $h$, Stegotext $t$, Payload length $|m^*|$, Maximum list length $2^N$;
    \ENSURE Recovered secret message $m^*$.
    \STATE $M \longleftarrow \{0,1\}^N$ \COMMENT{Initialize candidate message list}
    \STATE $l,cnt \longleftarrow N,0$ \COMMENT{Initialize counters}
    \STATE $suf \longleftarrow \mathsf{PRG}_{sk}(\cdot)$ \COMMENT{Reconstruct shared validation suffix}
    \WHILE{$l \leq |m^*| + |suf|$}
        \STATE $cnt \longleftarrow cnt + 1$
        \STATE $D \longleftarrow \mathsf{Model}(h)$ \COMMENT{Obtain distribution}
        \STATE $\{r_i\}_{i = 1}^{2|M|} \longleftarrow \mathsf{PRG}_{sk}(\cdot)$ \COMMENT{Reconstruct randomness}
        \STATE $\{s_i\}_{i = 1}^{|M|} \longleftarrow \mathsf{AliSample}_{|M|}\left(D,\{r\}_{i = 1}^{2|M|}\right)$ \COMMENT{Alias sampling}
        \FOR{$i=1$ {\bf to} $|M|$}
            \STATE Let $F(m_i) := s_i$ \COMMENT{Reconstruct mapping $F$}
        \ENDFOR
        \STATE $h \longleftarrow h \| t_{cnt}$ \COMMENT{Autoregressive update of history}
        \STATE $M \longleftarrow F^{-1}(t_{cnt})$ \COMMENT{Filter candidates by observed token}
        \WHILE{$|M| \leq 2^{N - 1}$}
            \STATE $M \longleftarrow \left\{m \| bit : m \in M, bit \in \{0,1\}\right\}$ \COMMENT{Expand candidate list}
            \STATE $l \longleftarrow l + 1$
        \ENDWHILE
    \ENDWHILE
    \STATE Find the unique $m \in M$ such that $m_{|m|-|suf|:|m|} = suf$
    \STATE Output $m^* \longleftarrow m_{1:|m^*|}$ \COMMENT{Drop validation suffix}
    \RETURN $m^*$
    \end{algorithmic}
\end{algorithm}

\subsection{Proof of Correctness}
\label{proof_of_correctness}
We now verify the condition under which the probability of non-unique decoding is negligible.
\begin{proposition}\label{prop3.1}
    For any non-secret message element $m$ in the candidate list, if the validation suffix $suf \in \{0,1\}^{b}$ is uniformly random (or pseudorandomly generated via $\mathsf{PRG}_{sk}$) and the maximum candidate list size is $2^N$, the probability that $m\|suf$ appears in the subsequent candidate list is at most $\frac{2^{-b}}{\left(1 + \sqrt{\frac{\lambda}{2^N}}\right)^n}$, where $n$ is the number of tokens added to embed $suf$.
\end{proposition}
\begin{proof}
    Let the additional generated token sequence of suffix be $s_1, s_2, \cdots, s_n$. In each step, the probability that an element $m$ in the candidate list is retained during filtering is $D(s_j)$. Given that $m$ exists in the candidate list, the probability of $m\|suf$ appearing is:
    \begin{align}
        \Pr\left[m\|suf \text{ appears} \mid m \text{ has appeared}\right] = \prod_{i = 1}^n D(s_j).
    \end{align}
    Since the candidate list expanded $b$ times during this process, we have:
    \begin{equation}\label{eq3.3}
        \prod_{i = 1}^n \left(D(s_i) + \delta \right)\leq 2^{-b},
    \end{equation}
    where $\delta$ is the deviation from the Hoeffding inequality. Equation \ref{eq3.3} can be expanded as:
    \begin{equation}
        \label{eq3.4}
        \prod_{i = 1}^n \left(D(s_i) + \delta \right)  = \prod_{i = 1}^n D(s_i) + \sum_{k=1}^n \delta^k \cdot \sum_{\substack{S \subseteq \{1,\dots,n\} \\ |S| = k}} \prod_{j \in S} D(s_j) \leq 2^{-b}.
    \end{equation}
    Since
    \begin{align}
       \sum_{k=1}^n \delta^k \cdot \sum_{\substack{S \subseteq \{1,\dots,n\} \\ |S| = k}} \prod_{j \in S} D(s_j) &\geq \sum_{k=1}^n  \binom{n}{k}\delta^k \prod_{j = 1}^n D(s_j) \notag\\&= \left((1 + \delta)^n - 1\right)\prod_{i = 1}^n D(s_i),
    \end{align}
    combining this with Eq. \ref{eq3.4}, we obtain:
    \begin{equation}
    \label{eq3.6}
        (1 + \delta)^n\prod_{i = 1}^n D(s_i) \leq 2^{-b}.
    \end{equation}
    Here we computes the maximum value of $\delta$: as the maximum candidate list size is $2^N$, we know that $|M_1|\geq 2^{N-1}$. By Hoeffding inequality, ensuring $\exp\left({-{2\delta^2\cdot \left(2^{N-1}\right)}}\right) \leq \exp\left(-\lambda\right)$ yields $\delta \leq \sqrt{\frac{\lambda}{2^N}}$. Using the conclusion of Eq. \ref{eq3.6}:
    \begin{align}
        \Pr\left[m\|suf \text{ appears} \mid m \text{ exists}\right]  = \prod_{i = 1}^n D(s_i) & \notag\\
        \leq \frac{2^{-b}}{(1 + \delta)^n} \leq \frac{2^{-b}}{\left(1 + \sqrt{\frac{\lambda}{2^N}}\right)^n}.&
    \end{align}
    By letting the right-hand side be less than $\exp(-\lambda)$, we maintain a negligible probability of collision when the suffix length $b$ satisfies:
    \begin{equation}
         b \geq \left\lceil \lambda\log_2e + n\log_2\left(1 + \sqrt{\frac{\lambda}{2^N}}\right) \right\rceil.
    \end{equation}

\end{proof}
That is, when the suffix $suf$ has length $b$, the probability that $m\|suf$ appears in the subsequent candidate-message set is at most $\frac{2^{-b}}{\left(1 + \sqrt{\frac{\lambda}{2^N}}\right)^n},$ where $n$ is the additional number of tokens generated in this process. Equivalently, if the suffix $suf$ has length at least $ \left\lceil \lambda\log_2 e \;+\; n\log_2\!\left(1 + \sqrt{\frac{\lambda}{2^N}}\right) \right\rceil,$ then the probability that $m\|suf$ appears in the subsequent candidate-message set is negligible.

\subsection{Proof of Security}
\label{security_proof}
\begin{figure}[htbp]
\begin{theobox}{$G_0$: Steganographic Encoding Algorithm}
    \begin{algorithmic}[1]
     \STATE $D \longleftarrow \mathsf{Model}(h)$\COMMENT{Obtain distribution}
    \STATE $\{r\}_{i = 1}^{2|M|} \longleftarrow \mathsf{PRG}_{sk}(\cdot)$\COMMENT{Generate pseudorandom numbers}
    \STATE {$\{s_i\}_{i = 1}^{|M|} \longleftarrow \mathsf{AliSample}_{|M|}\left(D,\{r\}_{i = 1}^{2|M|}\right)$ }\COMMENT{Alias sampling}
    \FOR{$i=1$ {\bf to} $|M|$}
        \STATE Let $F(m_i) := s_i$ \COMMENT{Construct mapping $F$}
    \ENDFOR
    \RETURN $F(m^*_{1:l})$ \COMMENT{Select sample $s^*$ corresponding to the real secret message}
    \end{algorithmic}
    \end{theobox}
\begin{theobox}{$G_1$: Using True Random Number}
    \begin{algorithmic}[1]
     \STATE $D \longleftarrow \mathsf{Model}(h)$\COMMENT{Obtain distribution}
    \STATE $\{r\}_{i = 1}^{2|M|} \longleftarrow_\$ \mathsf{Unif}[0,1] $\COMMENT{Get true random numbers}
    \STATE {$\{s_i\}_{i = 1}^{|M|} \longleftarrow \mathsf{AliSample}_{|M|}\left(D,\{r\}_{i = 1}^{2|M|}\right)$ }\COMMENT{Alias sampling}
    \FOR{$i=1$ {\bf to} $|M|$}
        \STATE Let $F(m_i) := s_i$ \COMMENT{Construct mapping $F$}
    \ENDFOR
    \RETURN $F(m^*_{1:l})$ \COMMENT{Select sample $s^*$ corresponding to the real secret message}
    \end{algorithmic}
    \end{theobox}
\begin{theobox}{$G_2$: Removing Secret Message}
    \begin{algorithmic}[1]
     \STATE $D \longleftarrow \mathsf{Model}(h)$\COMMENT{Obtain distribution}
    \STATE $\{r\}_{i = 1}^{2|M|} \longleftarrow_\$ \mathsf{Unif}[0,1] $\COMMENT{Get true random numbers}
    \STATE {$\{s_i\}_{i = 1}^{|M|} \longleftarrow \mathsf{AliSample}_{|M|}\left(D,\{r\}_{i = 1}^{2|M|}\right)$ }\COMMENT{Alias sampling}
    \STATE Choose $i \longleftarrow_\$ \{1,\dots,|M|\}$ \COMMENT{Uniformly select one of the samples}
    \STATE $s^* \longleftarrow s_i$
    \RETURN $s^*$
    \end{algorithmic}
    \end{theobox}
\begin{theobox}{$G_3$: Normal Model Sampling}
    \begin{algorithmic}[1]
     \STATE $D \longleftarrow \mathsf{Model}(h)$\COMMENT{Obtain distribution}
    \STATE $\{r\}_{i = 1}^{2} \longleftarrow_\$ \mathsf{Unif}[0,1] $\COMMENT{Get true random numbers}
    \STATE $s^* \longleftarrow \mathsf{AliSample}_{1}\left(D,\{r\}_{i = 1}^{2}\right)$ \COMMENT{Alias sampling}
    \RETURN $s^*$
    \end{algorithmic}
    \end{theobox}
    \caption{The games used in security proof.}
    \label{games}
\end{figure}

To prove the security of the proposed scheme, we construct the following sequence of games:
\begin{itemize}
    \item $G_0$: The steganography encoding algorithm $\mathsf{Encode}_{sk,\mathsf{Model},2^N}(h,m^*)$ using key-controlled pseudorandom numbers;
    \item $G_1$: $\mathsf{Encode}_{\mathsf{Model},2^N}(h,m^*)$ using truly random numbers;
    \item $G_2$: $\mathsf{Encode}_{\mathsf{Model},2^N}(h)$ using truly random numbers, with the secret message $m^*$ omitted;
    \item $G_3$: $\mathsf{Encode}_{\mathsf{Model},1}(h)$ using truly random numbers to draw a single sample via alias sampling, with the secret message $m^*$ omitted. The detailed procedure of $\mathsf{AliSample}$ is given in Appendix~\ref{alias-method};
\end{itemize}

We illustrate the games involved in the hybrid proof in Figure~\ref{games}. We can prove the following proposition:

\begin{proposition}
    For any history $h$, message $m^*$, secret key $sk$ generated by $\mathsf{KeyGen}$, the distribution of the stegotext generated by the encoding algorithm $\mathsf{Enc}_{sk,\mathsf{Model},2^N}(h,m^*)$ is computationally indistinguishable from the distribution of covertext generated by $\mathsf{Model}(h)$.
\end{proposition}

\begin{proof}
$G_0 \approx G_1$: Assume there exists a polynomial-time adversary $\mathcal{A}$ capable of effectively distinguishing the output of $G_0$ from that of $G_1$. We can then construct another polynomial-time adversary $\mathcal{A}'$ that can effectively distinguish between a true random number sequence and a pseudorandom number sequence. The workflow of $\mathcal{A}'$ is as follows:
\begin{itemize}
    \item Input a pseudorandom number sequence $\{r\}_{i = 1}^{2|M|}$ and a true random number sequence $\{r'\}_{i = 1}^{2|M|}$;
    \item $\mathcal{A}'$ runs $G_0$ using $\{r\}_{i = 1}^{2|M|}$ and $\{r'\}_{i = 1}^{2|M|}$ as the randomness, respectively, to obtain outputs $s$ and $s'$;
    \item $\mathcal{A}'$ invokes the adversary $\mathcal{A}$ to distinguish between $s$ and $s'$.
\end{itemize}
Since the $G_0$ algorithm using true random numbers is essentially identical to $G_1$, $s$ and $s'$ correspond to the outputs of $G_0$ and $G_1$ under the same history $h$ and secret message $m^*$, respectively. Thus, they can be effectively distinguished by the adversary $\mathcal{A}$. Consequently, $\mathcal{A}'$ can further distinguish the true random number sequence from the pseudorandom number sequence by invoking $\mathcal{A}$. However, by the definition of pseudorandom property, such an adversary $\mathcal{A}'$ does not exist. An adversary capable of distinguishing between algorithms $G_0$ and $G_1$ must necessarily be able to distinguish between true and pseudorandom sequences; therefore, such an adversary cannot operate in polynomial time. Since there is no polynomial-time adversary that can effectively distinguish between $G_0$ and $G_1$, by the definition of computational indistinguishability, the outputs of $G_0$ and $G_1$ are computationally indistinguishable.

$G_1 = G_2$: In $G_1$, the encoder outputs $s^* = F(m^*_{1:l})$, where $F$ is constructed by drawing i.i.d.\ samples from $D=\mathsf{Model}(h)$ using true randomness. Therefore $s^*$ is distributed exactly as a fresh sample from $D$, i.e., for any token $s$, $\Pr[s \leftarrow G_1(h,m^*)]=D(s)$. In $G_2$, the encoder instead chooses a uniformly random element from the same multiset of i.i.d.\ samples; by symmetry, the output distribution is also $D$. Hence $G_1$ and $G_2$ are identically distributed.

$G_2 = G_3$: $G_2$ randomly selects one element from the multiset of $|M|$ i.i.d.\ samples, while $G_3$ draws a single fresh sample from the distribution. It was proven in the previous step that the probability of $G_2$ outputting token $s$ is $D(s)$, following the distribution predicted by the model. Since $G_3$ utilizes alias sampling, the output of this sampling algorithm must follow the given input distribution $D$; thus, the probability of it outputting token $s$ remains $D(s)$.

In summary, we have established that $G_0 \approx G_1 = G_2 = G_3$. This demonstrates that the output of our steganographic encoding algorithm is computationally indistinguishable from the model's normal output at any single generation step. By the autoregressive nature of the generation process and the chain rule, it follows that the full text sequences generated by $\mathsf{Enc}_{sk,\mathsf{Model},2^N}(h,m)$ and $\mathsf{Model}(h)$ are computationally indistinguishable.

\end{proof}

\subsection{Capacity Analysis}
\label{capacity-analysis}
First, we define the embedding capacity within the steganographic scheme. In many prior schemes that recover payload information at the token level, capacity is typically defined as the average ratio of the number of bits that can be embedded and accurately extracted per token to the information of that token, which can be viewed as the information utilization rate. As our steganography scheme cannot decode any bit from a single token, the information utilization rate can be defined as follows:

\begin{definition}[Information Utilization Rate of steganographic scheme]
    Let the generated stegotext be a token sequence $s_1,s_2,\cdots,s_{n_{\mathrm{all}}}$. For each step $i$, let
    $p_i = \Pr[s_i \mid h_{<i}] = \mathsf{Model}(h_{<i})(s_i)$
    be the conditional probability assigned by the cover distribution. The total information content of the stegotext is
    \begin{align}
        I = \sum_{i = 1}^{n_{\mathrm{all}}} -\log_2(p_i).
    \end{align}
    Let $|m^*|$ denote the payload length (excluding the validation suffix). The information utilization rate is
    \begin{align}
        R = \frac{|m^*|}{I}.
    \end{align}
\end{definition}

Next, we prove the lower bound of the information utilization rate for the scheme; that is, the following proposition holds:

\begin{proposition}
    To embed a payload $m^*$, given that the maximum candidate list size is bounded by $2^N$ and the validation suffix length is $b$, a lower bound on the information utilization rate is
    $R \geq \left(1 - \frac{b - N}{|m^*|}\right)\cdot\left(1 - \frac{n_{\mathrm{all}}\sqrt{\frac{\lambda}{2^N}}}{\ln2 \cdot I}\right),$
    where $\lambda$ is the security parameter, $n_{\mathrm{all}}$ is the total number of generated tokens, and $I$ is the total information content of the stegotext.
\end{proposition}

\begin{proof}
    First, we determine the capacity lower bound on a single token. Without loss of generality, let the probability of the token be $p$. By Hoeffding's inequality, we have:
    \begin{align}
        \Pr\left[\frac{|M_2|}{|M_1|} - p\geq \delta p\right] \leq \exp\left(-{2^N\delta^2 p^2}\right).
    \end{align}
    By setting the right-hand side to be less than $\exp(-\lambda)$, rendering the probability of the event on the left-hand side negligible, we derive an upper bound for $\delta$ as $\sqrt{\frac{\lambda}{2^Np^2}}$. Therefore, after filtering based on a single token, the proportion of remaining candidate messages is at most $(1+\delta) p$.

    Next, consider the sequential generation of $n_{\mathrm{all}}$ tokens $s_1, s_2, \cdots, s_{n_{\mathrm{all}}}$. Let the probabilities of these tokens be $p_1, p_2, \cdots, p_{n_{\mathrm{all}}}$, respectively. In the worst-case scenario, if $\prod_{i = 1}^{n_{\mathrm{all}}} (1+\delta_i)p_i \leq \frac{1}{2}$, an expansion step is triggered. The probability of the product $\prod_{i = 1}^{n_{\mathrm{all}}} (1+\delta_i)p_i$ occurring is at most $\exp\left(-2^N\sum_{i = 1}^{n_{\mathrm{all}}} p_i^2\delta_i^2 \right)$. By constraining this probability upper bound to be less than $\exp(-\lambda)$, we obtain a constraint on $\delta_i$, namely $\sum_{i = 1}^{n_{\mathrm{all}}} p_i^2\delta_i^2 \leq \frac{\lambda}{2^N}$. Consequently, the information utilization rate implies:
    \begin{align}
        R = \frac{\sum_{i = 1}^{n_{\mathrm{all}}} -\log_2((1 + \delta_i)p_i)}{\sum_{i = 1}^{n_{\mathrm{all}}} -\log_2(p_i)}.
    \end{align}
    Considering the summation in the numerator, we can separate the product terms and apply the standard inequality $\ln(1+x)\leq x$:
    \begin{align}
        -\log_2((1 + \delta_i)p_i) =& -\log_2(p_i) - \frac{1}{\ln2}\ln(1 + \delta_i) \notag \\
        \geq & -\log_2(p_i) - \frac{\delta_i}{\ln 2}.
    \end{align}
    At this point, the information utilization rate $R$ simplifies to:
    \begin{align}
        R = 1 - \frac{\sum_{i = 1}^{n_{\mathrm{all}}} \frac{\delta_i}{\ln 2}}{\sum_{i = 1}^{n_{\mathrm{all}}} -\log_2(p_i)}.
    \end{align}
    Given that the constraint $\sum_{i = 1}^{n_{\mathrm{all}}} p_i^2\delta_i^2 \leq \frac{\lambda}{2^N}$ must hold for any $\delta_i$, and $0\leq p_i\leq 1$, it follows that $\sum_{i = 1}^{n_{\mathrm{all}}} \delta_i^2 \leq \frac{\lambda}{2^N}$. Viewing $\sum_{i = 1}^{n_{\mathrm{all}}} \delta_i^2$ as $\sum_{i = 1}^{n_{\mathrm{all}}} 1^2 \cdot \delta_i^2$ and applying the Cauchy-Schwarz inequality, we have:
    \begin{align}
        \left(\sum_{i = 1}^{n_{\mathrm{all}}} 1 \cdot \delta_i\right)^2 \leq \left(\sum_{i = 1}^{n_{\mathrm{all}}} 1^2\right) \left(\sum_{i = 1}^{n_{\mathrm{all}}} \delta^2 \right) \leq \frac{n_{\mathrm{all}}^2\lambda}{2^N}.
    \end{align}
    This implies $\sum_{i = 1}^{n_{\mathrm{all}}} {\delta_i} \leq n_{\mathrm{all}}\sqrt{\frac{\lambda}{2^N}}$. Denoting the total information content of tokens $s_1, s_2, \cdots s_{n_{\mathrm{all}}}$ as $I = \sum_{i =1 }^{n_{\mathrm{all}}} -\log_2(p_i)$, the lower bound for the information utilization rate is:
    \begin{align}
        R = 1 - \frac{\sum_{i = 1}^{n_{\mathrm{all}}} {\delta_i}}{\ln2 \cdot I} \geq 1 - \frac{n_{\mathrm{all}}\sqrt{\frac{\lambda}{2^N}}}{\ln2 \cdot I}.
    \end{align}
    Since $\sqrt{\frac{\lambda}{2^N}}$ is extremely small for typical choices of $N$ (e.g., $N\ge 20$), the overall utilization rate remains relatively close to 1.

    According to Proposition \ref{prop3.1}, unique decoding is achieved with high probability when the length of the embedded validation suffix $b$ is at least $\left\lceil \lambda\log_2e + n_{suf}\log_2\left(1 + \sqrt{\frac{\lambda}{2^N}}\right) \right\rceil$, where $n_{suf}$ is the number of tokens used to embed $suf$. In our overall encoding process, the payload is embedded starting from an initial $N$-bit prefix (because $|M_0|=2^N$) and ends after the required suffix bits are embedded. This introduces an overhead of
    \begin{equation}
    K = \left\lceil \lambda\log_2e + n_{suf}\log_2\left(1 + \sqrt{\frac{\lambda}{2^N}}\right) \right\rceil - N
     \end{equation}
    additional bits beyond the initial $N$-bit prefix. Therefore, the effective utilization rate is multiplied by a coefficient $1 - \frac{K}{|m^*|}$, which approaches $1$ as $|m^*|$ grows.
\end{proof}
    In conclusion, under the stated conditions, the lower bound of the information utilization rate for the steganography scheme is
    \begin{equation}
    R \geq \left(1 - \frac{K}{|m^*|}\right)\cdot\left(1 - \frac{n_{\mathrm{all}}\sqrt{\frac{\lambda}{2^N}}}{\ln2 \cdot I}\right).
    \end{equation}
   Hence, the lower bound on the average embedding capacity per token is
$\left(1 - \frac{K}{|m^*|}\right)\!\cdot\!\left(H - \frac{\sqrt{\lambda / 2^N}}{\ln 2}\right),$
which is asymptotically optimal with respect to the average entropy of tokens $H$.

\section{Experiment \& Discussion}

\subsection{Experiment Settings}

Although our steganographic scheme is not restricted to a specific carrier type, provably secure steganography has been most actively studied in text generation, as language models represent the best known technique for approximating human communication~\cite{meteor2021}. Therefore, to enable a more comprehensive and fair comparison with existing provably secure methods, our experiments primarily focus on large language model–based text generation scenarios.

\textbf{Baseline}.
We evaluated our steganography scheme against Random Sampling, representing the ideal steganography algorithm, and seven recent provably secure generative steganography algorithms, including METEOR~\cite{meteor2021} and its capacity-enhanced variant METEOR (R.), DISCOP~\cite{ding2023discop} and DISCOP (R.), SparSamp~\cite{wang2025sparsamp}, Shimmer~\cite{shimmer}, and the capacity-optimized group algorithm from FDPSS~\cite{liao2025framework}. Each baseline was introduced in Section~\ref{relatedworks} and tested using its default parameters.

\textbf{Model}.
We evaluate our method on three purely text-generation LLMs:\textsc{Mistral v0.3} \cite{mistral}, \textsc{Qwen2} \cite{qwen}, and \textsc{Llama3} \cite{llama3}, each with approximately 7 billion parameters. As models of this scale already exhibit strong instruction-following ability, all experiments are conducted in a zero-shot setting. In the main experiments, prompts are randomly sampled from the InstructWild dataset \cite{wild} and formatted with the tokenizer’s \texttt{apply\_chat\_template} function before being fed to the models. Because both our method and the provably secure baseline steganographic algorithms are distribution-agnostic and do not impose specific assumptions on the model distribution, we do not artificially constrain the decoding distribution. For each test, we ask the models to generate 1,000 samples using the full vocabulary, a base temperature of 1.0, and a maximum generation length of 800 tokens. In the main expirement, we set $N$=20 and use a 20-bit validation suffix (i.e., $b$=20), which targets a concrete decoding-failure upper bound on the order of $10^{-6}$.
All experiments were conducted on 2 $\times$ NVIDIA A5000 GPUs (32GB RAM) and 24 $\times$ Intel Xeon w5-3423 CPUs.
\begin{table*}[htb]
    \centering
     \resizebox{\textwidth}{!}{{
     \color{black}
    \begin{tabular}{cccccccccccc}
    \toprule[1.5pt]
     \multirow{2}{*}{\textbf{Model}}&\multirow{2}{*}{\textbf{Algorithm}}&  \multicolumn{3}{c}{\textbf{Capacity}}& \multicolumn{1}{c}{\textbf{Time}} & \multicolumn{4}{c}{\textbf{Linguistic Quality}}&\textbf{Imperceptibility}&\textbf{Correctness}\\
         \cmidrule(lr){3-5}
         \cmidrule(lr){6-7}
        \cmidrule(lr){7-10}
                \cmidrule(lr){11-11}
                  \cmidrule(lr){12-12}
         && \makecell{ Entropy\\ bit/token} &  \makecell{Embed \\ bit/token}$\uparrow$ &  Utili.$\uparrow$ &  \makecell{Time\\sec./token}$\downarrow$ &  PPL &  Dist2 &  Dist3 & Dist4 & TS-CSW F1$\downarrow$&SR$\uparrow$\\\midrule
\multirow{9}{*}{\textsc{Mistral}}
& Random Sampling & 0.8223 & - & - & 0.0552& 2.3359 & 0.2447&0.4067 & 0.5207& - & - \\
& METEOR (R.)~\cite{meteor2021}   & 0.8048     & 0.6737 & 0.6367 & 0.0552 & 2.2255& 0.2302& 0.3842 & 0.4906 & $52.39\%_{3.87}$ & 1.000 \\
& METEOR~\cite{meteor2021}   & 0.7794     & 0.6000 & 0.6253 & 0.0552& 2.3821 & 0.2374 & 0.3876& 0.4934 & $53.97\%_{2.99}$  & 1.000 \\
& DISCOP (R.)~\cite{ding2023discop}  & 0.7558 & 0.7439 & 0.8047& 0.0695 & 2.0709 & 0.2178 & 0.3541 & 0.4532 & $51.11\%_{2.37}$ & 1.000 \\
& DISCOP~\cite{ding2023discop}   & 0.7937      &0.4102 &0.4002 & 0.0606 & 2.1052& 0.2458 & 0.3899& 0.4860 & $50.08\%_{3.02}$ & 1.000 \\
& SparSamp~\cite{wang2025sparsamp}  & 0.7807       & 0.8100 & 0.8701 & 0.0544 & 2.3445 & 0.2366 & 0.3862 & 0.4854 & $51.10\%_{2.79}$ & 1.000 \\
& Shimmer~\cite{shimmer}  & 0.8163  &  0.6433& 0.7835 &0.0552& 2.0420 & 0.2014 & 0.3512& 0.4527 &$52.67\%_{2.35}$& 1.000\\
& Group~\cite{liao2025framework}  & 0.9827       & 0.6519 & 0.5875 & 0.0553 & 2.9868 & 0.2773& 0.4387& 0.5482 & $55.34\%_{3.94}$ & 1.000 \\
\cmidrule{2-12}
& Ours  & 0.8044  & 0.7927& \textbf{0.9835} & 0.0529 & 2.1121 & 0.2312& 0.3926 & 0.5013 & $50.06\%_{3.91}$& 1.000 \\

\midrule
 \multirow{9}{*}{\textsc{Qwen2}}
& Random Sampling  & 1.5591& - & - & 0.0546 & 3.5241 & 0.3445& 0.5635 & 0.7025& - & - \\
& METEOR (R.)~\cite{meteor2021}  & 1.6269  & 0.9634 & 0.5954 & 0.0546 & 3.7602& 0.3643& 0.5770 & 0.7061 & $52.69\%_{2.39}$ & 1.000 \\
& METEOR~\cite{meteor2021}   & 1.6197   & 0.9551 & 0.5933 & 0.0545 & 3.6935 & 0.3648 & 0.5827 & 0.7120 & $53.81\%_{2.90}$& 1.000 \\
& DISCOP (R.)~\cite{ding2023discop}  & 1.6201& 1.2342 & 0.8458 & 0.1296 &3.5792 & 0.3576 & 0.5753 & 0.7104 & $51.84\%_{3.19}$& 1.000 \\
& DISCOP~\cite{ding2023discop}  & 1.6231      & 0.7351 & 0.4396 & 0.1630 & 3.7863& 0.3632 & 0.5724& 0.6954 & $50.29\%_{2.83}$ & 1.000 \\
& SparSamp~\cite{wang2025sparsamp}  & 1.6238    & 1.3562& 0.9069 & 0.0539 & 3.6862 & 0.3634 & 0.5864 & 0.7241 & $50.77\%_{2.68}$ & 1.000 \\
& Shimmer~\cite{shimmer}  & 1.6734   & 1.1528 & 0.8403 & 0.0540 & 3.7181& 0.3852 & 0.5903& 0.7351& $52.71\%_{2.20}$&1.000\\
& Group~\cite{liao2025framework}  & 1.6611   & 1.0927 & 0.6197 & 0.0545 & 4.4333 & 0.3991 & 0.6151& 0.7431 & $53.59\%_{3.25}$ & 1.000 \\
\cmidrule{2-12}
& Ours & 1.5908 & 1.5971& \textbf{0.9906}& 0.0517& 3.8074 & 0.3602 & 0.5975 & 0.7237 & $50.59\%_{2.07}$& 1.000
\\
\midrule
 \multirow{9}{*}{\textsc{Llama3}}
& Random Sampling  & 0.6166& - & - & 0.0582 & 1.8595&   0.1846& 0.3143 & 0.4176 & - &-\\
& METEOR (R.)~\cite{meteor2021}  & 0.6724    & 0.3745 & 0.5418& 0.0587 & 1.8406 & 0.1800 & 0.3041& 0.4057 & $54.10\%_{3.15}$ & 1.000 \\
& METEOR~\cite{meteor2021}   & 0.6703 & 0.3640& 0.5406 & 0.0586 & 1.8525& 0.1893& 0.3210&0.4286 & $53.27\%_{2.72}$ & 1.000 \\
& DISCOP (R.)~\cite{ding2023discop}  & 0.6896    & 0.5135 & 0.7656 & 0.1093 & 1.8371 & 0.1816& 0.3073& 0.4109 & $51.94\%_{1.81}$& 1.000 \\
& DISCOP~\cite{ding2023discop}  & 0.6895        & 0.3376 & 0.4770 & 0.1376 & 1.9113 & 0.1842& 0.3121 & 0.4123& $50.43\%_{1.93}$ & 1.000 \\
& SparSamp~\cite{wang2025sparsamp}  & 0.6653      &  0.6265 & 0.9319 & 0.0575 & 1.8121 & 0.1803 & 0.3021& 0.4020& $51.92\%_{2.46}$&1.000\\
& Shimmer~\cite{shimmer}  & 0.7941 & 0.5733 & 0.7219& 0.0579& 2.0421& 0.2013 & 0.3413& 0.4463& $53.30\%_{3.08}$&1.000\\
& Group~\cite{liao2025framework}    & 0.7277    & 0.4364 & 0.5807 & 0.0584 & 2.1106 & 0.2122 & 0.3490& 0.4545 & $52.51\%_{3.20}$ & 1.000 \\
\cmidrule{2-12}
& Ours   & 0.6717  & 0.6725 & \textbf{0.9924} & 0.0523 & 1.8594 & 0.1893& 0.3154 & 0.4240& $51.21\%_{2.07}$ & 1.000
\\\bottomrule[1.5pt]

    \end{tabular} }}
    \caption{Main results with $N=20$ and $b=20$.}
    \label{tab:main}
\end{table*}

\subsection{Metrics}

We evaluate the performance of our steganography scheme in terms of \textbf{Security, Efficiency, Capacity and Correctness}.

\subsubsection{Security}
Although our steganography scheme's computational security has been theoretically proven in the Method Section, we further validate its security through linguistic steganalysis in adversarial scenarios and evaluate its intuitive imperceptibility based on linguistic quality.

\textbf{Steganalysis}. We employed a classic steganalysis model TS-CSW \cite{TS-CSW}, which is built upon the pre-trained BERT architecture \cite{devlin2018bert}to evaluate the imperceptibility of the generated stegotexts. The detectors were trained on 1,000 samples each of stegotext and corresponding cover text, randomly drawn from outputs of the same language model. The data were split in a 3:1:1 ratio for training, validation, and testing. Training was performed for three epochs with a learning rate of \(1 \times 10^{-4}\). This entire procedure was repeated three times, and the final steganalysis performance was reported as the average F1 score on the test set across all runs.

\textbf{Linguistic Quality.} The generated stegotext should meet the fundamental requirement of human perceptual concealment. We evaluate linguistic quality using perplexity (PPL) and diversity. PPL measures the fluency of generated text, with lower PPL values indicating smoother text:
\begin{equation}
PPL = \exp\left(-\frac{1}{L}\sum_{i=1}^{L}\log \ \Pr[x_i|{\bf x}_{1 : i - 1 }]\right),
\end{equation}
where $\bf x$ represents the generated text and $L$ is the token count of the generated text.

As for text diversity, we used the $ dist_n$ metric. This metric needs to find the unique pieces of tokens in the text and calculate their ratio:
\begin{equation}
    dist_n = \frac{\text{count(unique\, n-grams)}}{\text{count(n-grams)}}.
\end{equation}

\subsubsection{Efficiency}
\textbf{Time.}
To evaluate time efficiency, we measure the time required to generate a stego token, including both language-model inference and the additional computation introduced by the steganographic algorithm (e.g., sampling, and any auxiliary encoding/decoding operations). We report the average runtime per generated token over the full generation process to reduce variance. Lower time indicates higher efficiency.

\subsubsection{Capacity}

\textbf{Entropy.} This represents the theoretical upper bound of the embedding capacity, measured in bits per token. The entropy is computed as:
\begin{equation}
\text{Entropy} = -\frac{1}{L} \sum_{i=1}^{L} \Pr[x_i \mid x_{<i}] \log_2 \Pr[x_i \mid x_{<i}],
\end{equation}
where \(L\) denotes the text length, and \(\Pr[x_i \mid x_{<i}]\) is the conditional probability of token \(x_i\) given the preceding context \(x_{<i}\).

\textbf{Embedding Capacity.} This metric quantifies the average number of secret bits that can be successfully embedded and extracted, expressed as bits per token.

\textbf{Utilization Rate.} Let $B$ denote the number of secret bits that can be embedded and extracted (not including the suffix) in a sentence, and let
$ I = \sum_{i=1}^L -\log_2 \Pr[x_i \mid x_{<i}]$
denote the total information (i.e., the theoretical capacity in bits) of the generated sentence of length $L$. The utilization rate is defined as $R = \frac{B}{I}.$ Here $R \in [0,1]$, and $R = 1$ indicates optimal use of the available entropy.

\subsubsection{Correctness}
\textbf{Success Rate}
We evaluate correctness using the bit-level decoding Success Rate (SR).
Let $\mathbf{x}$ denote the original secret bitstream and $\hat{\mathbf{x}}$ the decoded bitstream, both of length$ |\mathbf{x}|$. We define $SR = \frac{1}{|\mathbf{x}|}\sum_{j=1}^{|\mathbf{x}|} \mathbf{\delta}[x_j , \hat{x}_j],$
where $\mathbf{\delta}[u,v]$ is the indicator function, which takes the value 1 if $u=v$ and 0 otherwise.  Consequently, $SR = 1$ signifies a lossless reconstruction of the secret message, whereas lower values indicate the presence of decoding errors.

\subsection{Evaluation Results and Discussion}

\begin{figure*}[htbp]
	\centering
	\includegraphics[width=\textwidth]{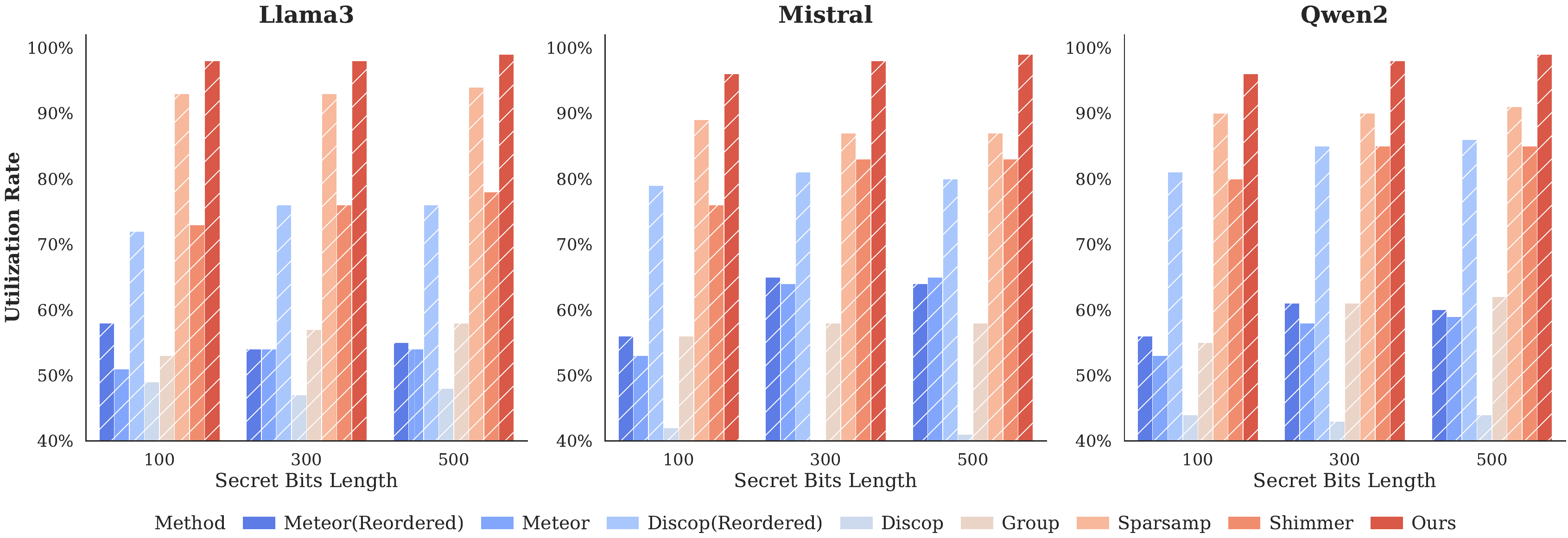}
	\caption{The information utilization rate of several existing provable secure steganography schemes}
	\label{utils}
\end{figure*}

\subsubsection{Capacity}
As analyzed in Section~\ref{capacity-analysis}, we derive lower bounds on both the information utilization rate and the embedding capacity. In the main experimental setting ($b=20$, $N=20$, hence $K=0$), the lower bound on the average embedding capacity per token is
$H - \frac{\sqrt{\lambda / 2^N}}{\ln 2}.$
 This bound exceeds the explicit lower bounds reported by GROUP~\cite{liao2025framework} ($H - \log_2(1 + H \ln 2) - 0.0861$), METEOR~\cite{meteor2021} ($\tfrac{1}{2}H - \tfrac{1}{2}$), DISCOP~\cite{ding2023discop} (minimum entropy), and Shimmer~\cite{shimmer} ($\tfrac{2}{3}H - \tfrac{1}{5}$) in LLM-based experiments, where the entropy H is typically small.
Although some recent SOTA provably secure steganography methods, such as SparSamp~\cite{wang2025sparsamp}, do not provide theoretical capacity analyses, the main experimental results in Table~\ref{tab:main} show that our scheme achieves higher entropy utilization than SparSamp~\cite{wang2025sparsamp}.

We further compare information utilization rate across large language models and secret message lengths (Fig.~\ref{utils}). A consistent trend is that Shimmer~\cite{shimmer} and SparSamp~\cite{wang2025sparsamp} achieve higher capacity than METEOR~\cite{meteor2021}, METEOR (R).~\cite{meteor2021}, DISCOP~\cite{ding2023discop}, DISCOP (R.)~\cite{ding2023discop}, and Group~\cite{liao2025framework}. This gap is structural: METEOR, DISCOP, and Group recover payload information from the current token and the current-step distribution only. In regimes where the next-token distribution is highly peaked, as is typical for LLMs, multiple payload candidates often map to the same high-probability token. Consequently, these methods must either embed only a short common prefix or forgo embedding entirely, thereby wasting part of the available entropy. Shimmer and SparSamp, by contrast, accumulate decoding evidence across multiple steps, which allows them to exploit entropy more effectively. Among the token-local baselines, DISCOP(R.) achieves the highest capacity in our experimental settings~\cite{wang2025sparsamp,shimmer}, but it still falls short of the best multi-step designs.

Despite their advantage, prior high-capacity schemes still incur avoidable capacity loss due to how they represent and update the decoding state. Shimmer~\cite{shimmer} maintains a single interval-valued state; when interval splitting occurs, the valid set becomes a union of disjoint intervals. Keeping all branches would lead to an exponential blowup, so Shimmer must effectively discard branches to suppress splitting, which wastes entropy and reduces utilization. SparSamp~\cite{wang2025sparsamp} improves capacity by accumulating entropy over blocks, but it does not expand a set of parallel decoding candidates when ambiguity persists; instead, it must pause active embedding (or spend additional tokens to resolve the state), again leaving part of the entropy unexploited.

In contrast, our method attains the highest capacity among the compared baselines because it adopts a list-decoding view of steganographic decoding. Rather than enforcing immediate decodability or committing to a single interval/state, we maintain a bounded candidate list of message prefixes. Each generated token filters this list in proportion to its probability mass, and whenever the list becomes too small we expand it to keep the filtering process operating near the entropy rate. The remaining ambiguity is then eliminated by matching a short validation suffix. This filter--expand--match mechanism continuously harvests entropy from every token while keeping computation bounded, which explains the consistent utilization gains observed in Fig.~\ref{utils}

We further examine how the secret message length $|m^*|$ and the candidate list size $2^N$ affect the utilization rate (Fig.~\ref{fig:heatmap-figs}). Across all three large language models, for a fixed $N$, longer secret messages yield higher utilization. This trend is expected because the suffix can be viewed as a fixed overhead used to enable fast and unambiguous decoding; amortizing the same suffix over more payload bits improves overall efficiency (even though suffix bits are excluded from the embedding capacity).
Moreover, larger values of $N$ consistently yield higher utilization, which aligns with our list-decoding intuition: maintaining a larger candidate list provides more flexibility and results in higher embedding capacity.

\begin{figure}[htbp]
    \centering

    \begin{subfigure}[b]{0.99\linewidth}
        \centering
        \includegraphics[width=0.98\linewidth]{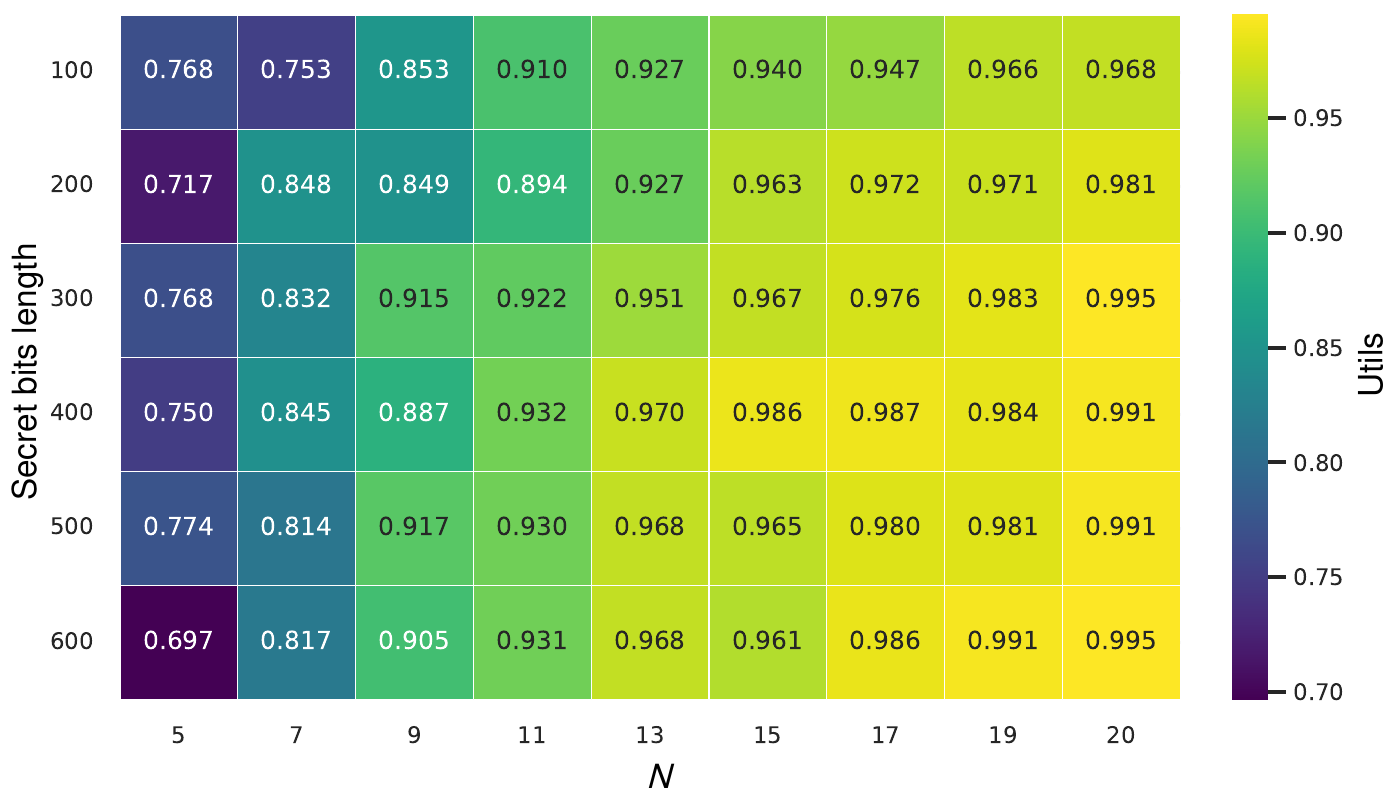}
        \caption{\textit{Llama3}}
        \label{fig:subfig1}
    \end{subfigure}

    \begin{subfigure}[b]{0.99\linewidth}
        \centering
        \includegraphics[width=0.98\linewidth]{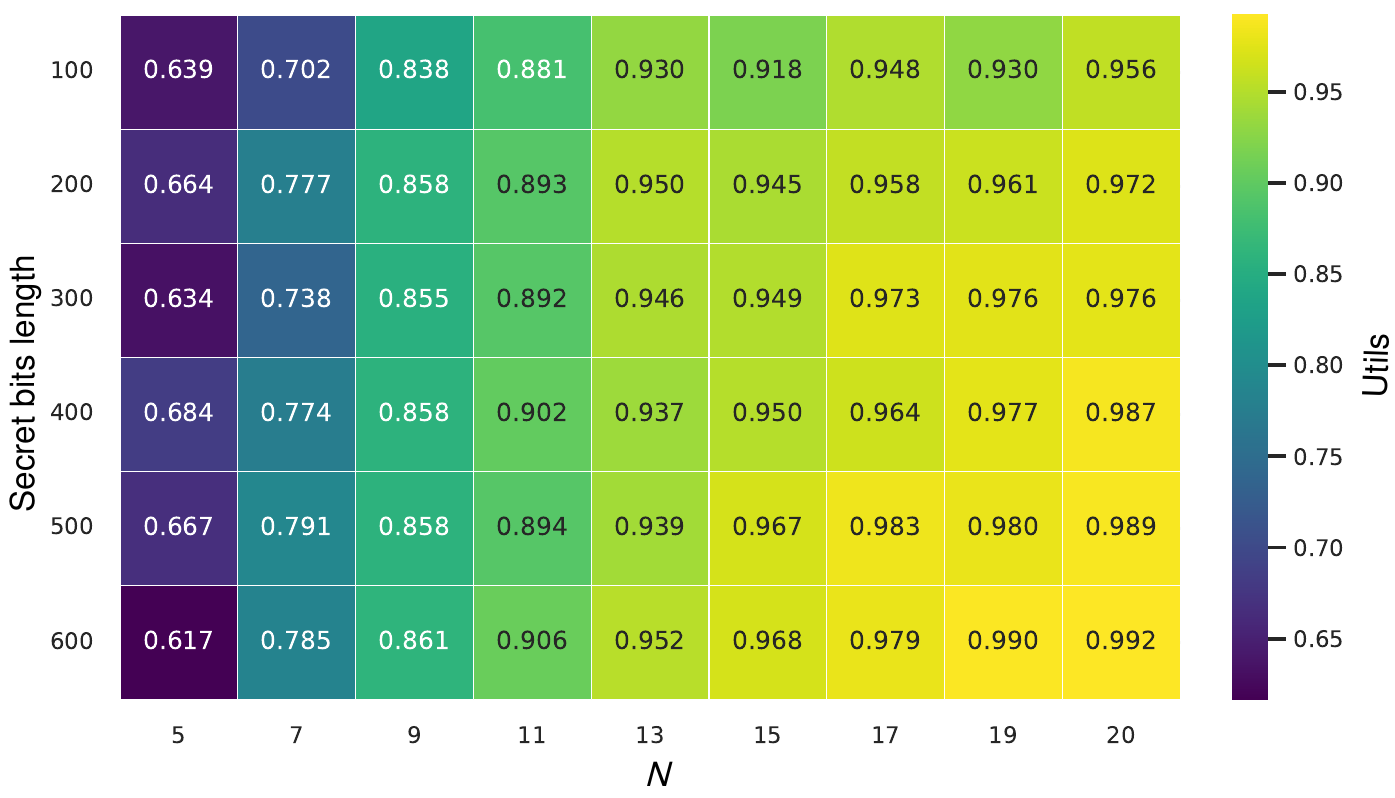}
        \caption{\textit{Mistral3}}
        \label{fig:subfig2}
    \end{subfigure}

    \begin{subfigure}[b]{0.99\linewidth}
        \centering
        \includegraphics[width=0.98\linewidth]{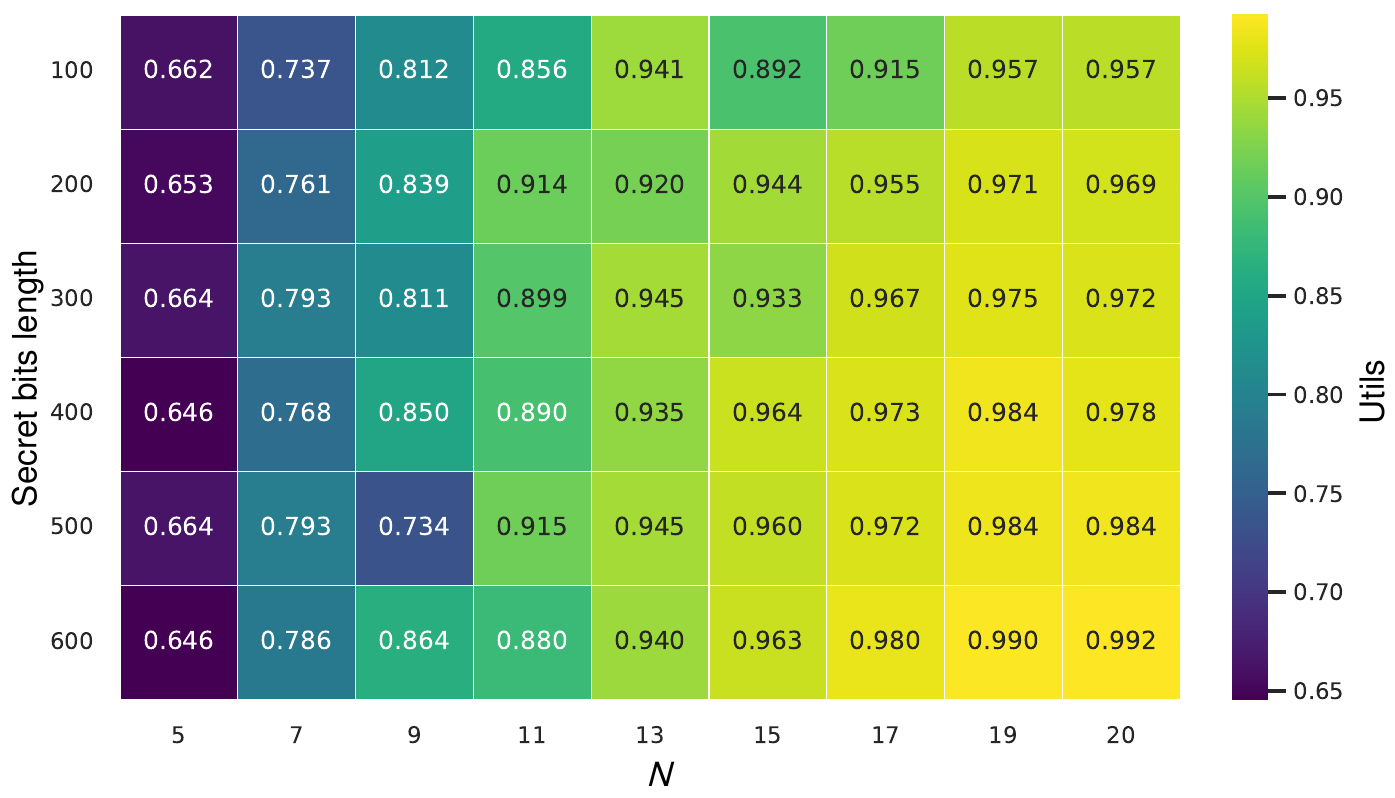}
        \caption{\textit{Qwen2}}
        \label{fig:subfig3}
    \end{subfigure}

    \caption{Information utilization rate of our method for different secret bit lengths and candidate list length ($N$) on three LLMs (\textit{Llama3}, \textit{Mistral3}, and \textit{Qwen2}).}
    \label{fig:heatmap-figs}
\end{figure}

\subsubsection{Time Efficiency}
The time of our approach is predominantly limited by two factors:  the time for model to generate token and the time for sampling. As shown in Table~\ref{tab:main}, even with a sampling scale of around $2^{20}$, our steganography scheme does not incur significant time overhead; its runtime is comparable to random sampling and other baseline methods. With faster sampling techniques~\cite{alisa1,alisa2}, the time required to generate stego tokens could be further reduced. Our scheme is even faster than DISCOP, especially its reordered variant, which must construct a Huffman tree at each time step, which leads to slightly higher runtime despite our optimized Cython implementation.

\subsubsection{Linguistic Quality \& Security }

Similar to other provably secure steganographic constructions, Table~\ref{tab:main} shows that the stegotext generated by our method closely matches randomly sampled text across multiple evaluation metrics, including PPL and diversity measures (Dist-2, Dist-3, and Dist-4). This is consistent with our security analysis, indicating that the stegotext is difficult to distinguish from cover text. As illustrated in Table~\ref{casestudy}, both randomly sampled text and stego text are also perceptually indistinguishable to human readers. Steganalysis results further confirm this: because provably secure schemes make steganographic sampling computationally indistinguishable from random sampling, machine-learning classifiers fail to learn effective features, yielding F1 scores close to 50\%.

\subsubsection{Decoding Correctness}
\label{subsec:correctness}
 In the Proof of Correctness section~\ref{proof_of_correctness}, we explain that when the suffix \( suf \) has a length of \( b \), the probability of a decoding error is at most \(\frac{2^{-b}}{\left(1 + \sqrt{\frac{\lambda}{2^N}}\right)^n}\). Larger \(b\) tightens this bound and can make the decoding error probability negligibly small; in our experiments we choose \( b = 20 \) to reduce overhead while still keeping the concrete upper bound extremely low. In the main experiment, where \( b = 20 \), the upper bound for the error rate is \( 8.695 \times 10^{-7} \). As indicated in Table~\ref{tab:main}, no decoding errors were observed in the main experiments involving 1,000 samples.

\section{Conclusion}
In this paper, we revisit steganography decoding from a list-decoding perspective and show that temporarily maintaining a set of decoding candidates can exploit entropy more fully than prior token-local designs. Building on this insight, we propose a provably secure steganography scheme based on list decoding, which combines list expansion and suffix matching to ensure correct decoding while enabling high-capacity embedding. We provide a security and correctness proof,and derive theoretical lower bound on the embedding capacity.
Extensive experiments on three popular LLMs demonstrate that our method achieves high generation quality, efficiency, and capacity.

\section{Acknowledgement}
The authors gratefully acknowledge Liao, Shao, et al. for sharing their perspective on classification of instantaneous and non-instantaneous steganography in their ongoing work.
\section{Ethics Considerations}

We propose a provably secure steganography scheme based on list decoding that embeds secret messages into seemingly innocuous text. We acknowledge the \textbf{dual-use risk} of this technology: While it protects people's privacy to prevent censorship, it could also be exploited by malicious actors for covert coordination. We treat the secret bits as a random binary bitstream without semantic constraints in expirements so the steganography method is agnostic to content.
Relative to prior provably secure text steganography~\cite{meteor2021,ding2023discop,shimmer,liao2025framework}, our main contribution is an algorithmic design that improves capacity and decoding, rather than the introduction of the underlying ethically sensitive capability of covert communication itself. We therefore do not believe that this work materially exacerbates the ethical risks already associated with this line of research.

Furthermore, we emphasize that our scheme guarantees concealment only at the \textbf{content level}. It does not hide network metadata or traffic patterns, which monitors can still exploit via behavioral analysis. Thus, this method should be viewed not as a standalone panacea for resisting surveillance but as a component of a broader surveillance defense strategy.

\newpage

\newpage

\appendix
\section{Appendix}
\subsection{Alias Method}
\label{alias-method}
Since the number of samples required can be extremely large, the efficiency of standard Inverse Transform Sampling is insufficient. The complexity of traditional Inverse Transform Sampling is $O(|V|)$, where $|V|$ is the size of the model's vocabulary. If $2^k$ samples are required, the complexity of the entire sampling flow becomes $O(2^k|V|)$. In practice, generating 1 million samples could incur a latency of several seconds or even tens of seconds. Therefore, we introduce an efficient sampling method known as the \textbf{Alias Method}\cite{alisa1,alisa2}, proposed by A. J. Walker in 1977. The detailed procedure is given below.

\begin{algorithm}[h]
    \caption{Alias Sampling Method $\mathsf{AliSample}_N(D,\{r\}_{i = 1}^{2N}) \xrightarrow[]{} \{s\}_{i = 1}^N$}
    \label{alias-sampling}
    \renewcommand{\algorithmicrequire}{\textbf{Input:}}
    \renewcommand{\algorithmicensure}{\textbf{Output:}}
    \begin{algorithmic}[1]
    \REQUIRE Distribution $D$, Sequence of random numbers $\{r\}_{i = 1}^{2N}$;
    \ENSURE Sequence of samples $\{s\}_{i = 1}^N$.
        \STATE $S, G \longleftarrow \emptyset, \emptyset$ \COMMENT{Preparation phase: Initialize the list of small probability tokens $S$ and large probability tokens $G$.}
        \FOR{$i \in \{1,2,\cdots,|V|\}$}
            \STATE $q_i \longleftarrow |V|\cdot D(i)$ \COMMENT{Scale probability by a factor of $|V|$}
            \IF{$q_i < 1$}
                \STATE $S \longleftarrow S + \{i\}$ \COMMENT{Construct low-probability set $S$}
            \ELSE
                \STATE $G \longleftarrow G + \{i\}$\COMMENT{Construct high-probability set $G$}
            \ENDIF
        \ENDFOR
        \FOR{$i \in \{1,2,\cdots,|V|\}$}
            \STATE Select and remove an element $j$ from set $S$
            \STATE $A_i, P(A_i) \longleftarrow j, q_j$ \COMMENT{Prioritize low-probability tokens to construct sub-distribution}
            \IF{$G \neq \emptyset$}
                \STATE Select and remove an element $k$ from set $G$
                \STATE $B_i, P(B_i) \longleftarrow k, 1 - q_j$ \COMMENT{Select a high-probability token to fill the remainder}
                \STATE $q_k \longleftarrow q_k - (1 - q_j)$\COMMENT{Update token probability}
                \IF{$q_k < 1$}
                    \STATE $S \longleftarrow S + \{k\}$ \COMMENT{If probability drops below 1, move to low-probability set}
                \ELSE
                    \STATE $G \longleftarrow G + \{k\}$\COMMENT{If probability remains $\ge 1$, keep in high-probability set}
                \ENDIF
            \ENDIF
        \ENDFOR
        \FOR{$i \in \{1,2,\cdots,N\}$}
            \STATE $x \longleftarrow \lfloor |V| \cdot r_{2i -1} \rfloor$ \COMMENT{Randomly select one of the $|V|$ uniform sub-distributions}
            \IF{$r_{2i} < P(A_x)$}
                \STATE $s_i \longleftarrow A_x$ \COMMENT{Sample from the binary sub-distribution; repeat $N$ times}
            \ELSE
                \STATE $s_i \longleftarrow B_x$
            \ENDIF
        \ENDFOR
        \RETURN $\{s_i\}_{i = 1}^N$
    \end{algorithmic}
\end{algorithm}

The algorithm reconstructs the original distribution $D$ in a preparation phase into a uniform combination of sub-distributions, each containing at most two tokens, such that $D(z) = \sum_{i = 1}^{|V|} \frac{1}{|V|} D_i(z)$. To construct such sub-distributions, the probability of each token is first multiplied by $|V|$ as a weight. These weights are then redistributed into multiple normal sub-distributions that sum to 1. Specifically, tokens are divided into two categories: a low-weight list $S$ (weights $<1$) and a high-weight list $G$ (weights $\ge 1$). When constructing a sub-distribution, a token $j$ is removed from $S$. Since $q_j < 1$, a token $k$ is removed from $G$, and its weight $q_k$ is split. Part of it, $1-q_j$, is used to fill the sub-distribution (making the sum 1), while the remainder $q_k - (1-q_j)$ becomes the new weight for token $k$, which is then re-added to either $S$ or $G$. The resulting sub-distribution contains only two tokens, $A_i = j$ and $B_i = k$, with probabilities $P(A_i) = q_j$ and $P(B_i) = 1 - q_j$.

Eventually, the entire distribution $D$ is decomposed into a uniform combination of such sub-distri-\\butions. Once constructed, sampling is straightforward: uniformly select one of the $|V|$ sub-distributions, then output the sample based on a comparison between a random number and $P(A_i)$. The complexity of the preparation phase is $O(|V|)$, while a single sampling operation requires only one comparison, yielding $O(1)$ complexity. When generating $2^k$ samples, the total complexity is $O(|V| + 2^k)$, which is significantly superior to repeated Inverse Transform Sampling.

\subsection{Examples of generated stegotext}

Table~\ref{casestudy} shows several examples of stegotext produced by \textit{llama3} in the main experiment. It is difficult for human to perceptually distinguish between normally generated text produced by random sampling and text generated by provably secure steganographic methods.

\begin{table*}[htbp]
    \centering
    \caption{Examples of stegotexts generated by \textit{Llama3-8B}}
     \resizebox{\textwidth}{!}{
     {
    \begin{tabularx}{\textwidth}{>{\raggedright\arraybackslash}p{3cm}|X}
        \toprule[1.5pt]

        \textbf{Method} & \textbf{Generated Texts} \\
        \midrule
        Prompt & What is Steganography? \\
        \midrule
        Random Sampling &Steganography is a technique of hiding secret information within a non-secret message, image, audio, or other medium. The goal is to conceal the existence of the secret information, making it difficult for unauthorized parties to detect its presence. Imagine sending a letter to a friend with a hidden message inside: you would write the secret message on the back of a seemingly ordinary letter or enclose it within an envelope with a false address...... \\\midrule
        METEOR(R.)~\cite{meteor2021}     &  Steganography is the practice of hiding secret information within a non-secret message, image, audio, or other medium. The goal of steganography is to conceal the existence of the secret information, making it difficult for unauthorized parties to detect its presence. In other words, steganography involves embedding a message, known as the covertext, within a seemingly innocuous medium, such as an image, video, audio file, or text......   \\\midrule
		           METEOR~\cite{meteor2021}    &  Steganography is the practice of hiding secret information within a non-secret message, image, or file in such a way that the very existence of the secret information is not apparent. The goal of steganography is to conceal the fact that a secret message is being transmitted, rather than simply encrypting it. In other words, steganography is a technique used to communicate sensitive information securely by hiding it in plain sight......\\\midrule
		           DISCOP(R.)~\cite{ding2023discop}      &   Steganography is a technique used to hide secret information within a non-secret message, image, audio, or video file. The goal of steganography is to conceal the existence of the secret information, making it difficult for unauthorized parties to detect its presence. In steganography, a cover text is used as a carrier to hide the secret information. The cover text can be a digital image, audio file, video, or even a text file. The secret information, known as the message, is then concealed ...... \\\midrule
             DISCOP~\cite{ding2023discop} & Steganography is a method of hiding secret information within a non-secret message, image, audio, or video file. The idea is to conceal the existence of the hidden message or data, making it difficult for unauthorized parties to detect its presence. In other words, steganography is a technique that embeds a secret message or data within a cover file, such as an image, audio, or video, in a way that the presence of the hidden information is not immediately apparent...\\\midrule
		              SparSamp~\cite{wang2025sparsamp}  &  Steganography is the practice of hiding secret information within a non-secret message, image, audio, or video file in such a way that the very existence of the secret message is not apparent. The goal of steganography is to conceal the presence of the secret information, making it invisible to unauthorized parties. Steganography is often compared to cryptography, which is the practice of encrypting information to make it secure. While cryptography focuses on protecting the confidentiality......\\\midrule
                Shimmer~\cite{shimmer} &Steganography is the practice of hiding secret information within a non-secret message, image, audio file, or other medium in such a way that the very existence of the secret information is not apparent. This is in contrast to cryptography, which encrypts the secret information, making it unreadable without the proper decryption key or access. In steganography, the secret information is embedded in a way that makes it difficult to distinguish from the rest of the message, and the presence of......\\\midrule
                Group~\cite{liao2025framework} & Steganography is the practice of concealing a secret message, image, or audio file within a non-secret message, image, or audio file, called a cover object. The goal of steganography is to hide the existence of the secret message orsignal, so that only the intended recipient can detect and extract it. Steganography is often used to achieve confidentiality and authenticity of the secret information. By embedding the secret message within a seemingly innocuous object......\\\midrule
                Ours & Steganography is the practice of hiding secret information or messages within another, seemingly innocuous, medium or format, such as an image, audio file, or text document. The purpose of steganography is to conceal the existence of the secret message, making it difficult for unauthorized parties to detect or intercept it. Steganography is different from cryptography, which focuses on encrypting and protecting the content of the message. In steganography, the message itself ...... \\

        \bottomrule[1.5pt]
    \end{tabularx}}}
    \label{casestudy}
\end{table*}


\begin{thebibliography}{20}


\ifx \showCODEN    \undefined \def \showCODEN     #1{\unskip}     \fi
\ifx \showDOI      \undefined \def \showDOI       #1{#1}\fi
\ifx \showISBNx    \undefined \def \showISBNx     #1{\unskip}     \fi
\ifx \showISBNxiii \undefined \def \showISBNxiii  #1{\unskip}     \fi
\ifx \showISSN     \undefined \def \showISSN      #1{\unskip}     \fi
\ifx \showLCCN     \undefined \def \showLCCN      #1{\unskip}     \fi
\ifx \shownote     \undefined \def \shownote      #1{#1}          \fi
\ifx \showarticletitle \undefined \def \showarticletitle #1{#1}   \fi
\ifx \showURL      \undefined \def \showURL       {\relax}        \fi
\providecommand\bibfield[2]{#2}
\providecommand\bibinfo[2]{#2}
\providecommand\natexlab[1]{#1}
\providecommand\showeprint[2][]{arXiv:#2}

\bibitem[Bai et~al\mbox{.}(2023)]%
        {qwen}
\bibfield{author}{\bibinfo{person}{Jinze Bai}, \bibinfo{person}{Shuai Bai},
  {and} \bibinfo{person}{Yunfei~Chu et. al}.} \bibinfo{year}{2023}\natexlab{}.
\newblock \showarticletitle{Qwen Technical Report}.
\newblock \bibinfo{journal}{\emph{arXiv preprint arXiv:2309.16609}}
  (\bibinfo{year}{2023}).
\newblock


\bibitem[Bai et~al\mbox{.}(2025)]%
        {shimmer}
\bibfield{author}{\bibinfo{person}{Minhao Bai}, \bibinfo{person}{Kaiyi Pang},
  \bibinfo{person}{Guorui Liao}, \bibinfo{person}{Jinshuai Yang}, {and}
  \bibinfo{person}{Yongfeng Huang}.} \bibinfo{year}{2025}\natexlab{}.
\newblock \showarticletitle{Shimmer: a Provably Secure Steganography Based on
  Entropy Collecting Mechanism}. In \bibinfo{booktitle}{\emph{34th USENIX
  Security Symposium (USENIX Security 25)}}. \bibinfo{pages}{5949--5965}.
\newblock


\bibitem[Devlin et~al\mbox{.}(2018)]%
        {devlin2018bert}
\bibfield{author}{\bibinfo{person}{Jacob Devlin}, \bibinfo{person}{Ming-Wei
  Chang}, \bibinfo{person}{Kenton Lee}, {and} \bibinfo{person}{Kristina
  Toutanova}.} \bibinfo{year}{2018}\natexlab{}.
\newblock \showarticletitle{Bert: Pre-training of deep bidirectional
  transformers for language understanding}.
\newblock \bibinfo{journal}{\emph{arXiv preprint arXiv:1810.04805}}
  (\bibinfo{year}{2018}).
\newblock


\bibitem[Ding et~al\mbox{.}(2023)]%
        {ding2023discop}
\bibfield{author}{\bibinfo{person}{Jinyang Ding}, \bibinfo{person}{Kejiang
  Chen}, \bibinfo{person}{Yaofei Wang}, \bibinfo{person}{Na Zhao},
  \bibinfo{person}{Weiming Zhang}, {and} \bibinfo{person}{Nenghai Yu}.}
  \bibinfo{year}{2023}\natexlab{}.
\newblock \showarticletitle{Discop: Provably Secure Steganography in Practice
  Based on “Distribution Copies”}. In \bibinfo{booktitle}{\emph{2023 IEEE
  Symposium on Security and Privacy (SP)}}. IEEE Computer Society,
  \bibinfo{pages}{2238--2255}.
\newblock


\bibitem[Dubey et~al\mbox{.}(2024)]%
        {llama3}
\bibfield{author}{\bibinfo{person}{Abhimanyu Dubey}, \bibinfo{person}{Abhinav
  Jauhri}, \bibinfo{person}{Abhinav Pandey}, \bibinfo{person}{Abhishek Kadian},
  \bibinfo{person}{Ahmad Al-Dahle}, \bibinfo{person}{Aiesha Letman},
  \bibinfo{person}{Akhil Mathur}, \bibinfo{person}{Alan Schelten},
  \bibinfo{person}{Amy Yang}, \bibinfo{person}{Angela Fan}, {et~al\mbox{.}}}
  \bibinfo{year}{2024}\natexlab{}.
\newblock \showarticletitle{The llama 3 herd of models}.
\newblock \bibinfo{journal}{\emph{arXiv preprint arXiv:2407.21783}}
  (\bibinfo{year}{2024}).
\newblock


\bibitem[Fang et~al\mbox{.}(2017)]%
        {fang2017generating}
\bibfield{author}{\bibinfo{person}{Tina Fang}, \bibinfo{person}{Martin Jaggi},
  {and} \bibinfo{person}{Katerina Argyraki}.} \bibinfo{year}{2017}\natexlab{}.
\newblock \showarticletitle{Generating Steganographic Text with LSTMs}. In
  \bibinfo{booktitle}{\emph{Proceedings of ACL 2017, Student Research
  Workshop}}. \bibinfo{pages}{100--106}.
\newblock


\bibitem[Guruswami and Sudan(1999)]%
        {listdecoding1}
\bibfield{author}{\bibinfo{person}{V. Guruswami} {and} \bibinfo{person}{M.
  Sudan}.} \bibinfo{year}{1999}\natexlab{}.
\newblock \showarticletitle{Improved decoding of Reed-Solomon and
  algebraic-geometry codes}.
\newblock \bibinfo{journal}{\emph{IEEE Transactions on Information Theory}}
  \bibinfo{volume}{45}, \bibinfo{number}{6} (\bibinfo{year}{1999}),
  \bibinfo{pages}{1757--1767}.
\newblock
\urldef\tempurl%
\url{https://doi.org/10.1109/18.782097}
\showDOI{\tempurl}


\bibitem[Hopper(2004)]%
        {hopper2004toward}
\bibfield{author}{\bibinfo{person}{Nicholas~J. Hopper}.}
  \bibinfo{year}{2004}\natexlab{}.
\newblock \emph{\bibinfo{title}{Toward a Theory of Steganography}}.
\newblock Ph.D. thesis. \bibinfo{school}{Carnegie Mellon University},
  \bibinfo{address}{Pittsburgh, PA}.
\newblock


\bibitem[Jiang et~al\mbox{.}(2023)]%
        {mistral}
\bibfield{author}{\bibinfo{person}{Albert~Q Jiang}, \bibinfo{person}{Alexandre
  Sablayrolles}, \bibinfo{person}{Arthur Mensch}, \bibinfo{person}{Chris
  Bamford}, \bibinfo{person}{Devendra~Singh Chaplot}, \bibinfo{person}{Diego
  de~las Casas}, \bibinfo{person}{Florian Bressand}, \bibinfo{person}{Gianna
  Lengyel}, \bibinfo{person}{Guillaume Lample}, \bibinfo{person}{Lucile
  Saulnier}, {et~al\mbox{.}}} \bibinfo{year}{2023}\natexlab{}.
\newblock \showarticletitle{Mistral 7B}.
\newblock \bibinfo{journal}{\emph{arXiv preprint arXiv:2310.06825}}
  (\bibinfo{year}{2023}).
\newblock


\bibitem[Kaptchuk et~al\mbox{.}(2021)]%
        {meteor2021}
\bibfield{author}{\bibinfo{person}{Gabriel Kaptchuk},
  \bibinfo{person}{Tushar~M. Jois}, \bibinfo{person}{Matthew Green}, {and}
  \bibinfo{person}{Aviel~D. Rubin}.} \bibinfo{year}{2021}\natexlab{}.
\newblock \showarticletitle{Meteor: Cryptographically Secure Steganography for
  Realistic Distributions}. In \bibinfo{booktitle}{\emph{Proceedings of the
  2021 ACM SIGSAC Conference on Computer and Communications Security}} (Virtual
  Event, Republic of Korea) \emph{(\bibinfo{series}{CCS '21})}.
  \bibinfo{publisher}{Association for Computing Machinery},
  \bibinfo{address}{New York, NY, USA}, \bibinfo{pages}{1529–1548}.
\newblock
\showISBNx{9781450384544}
\urldef\tempurl%
\url{https://doi.org/10.1145/3460120.3484550}
\showDOI{\tempurl}


\bibitem[Liao et~al\mbox{.}(2025)]%
        {liao2025framework}
\bibfield{author}{\bibinfo{person}{Guorui Liao}, \bibinfo{person}{Jinshuai
  Yang}, \bibinfo{person}{Weizhi Shao}, {and} \bibinfo{person}{Yongfeng
  Huang}.} \bibinfo{year}{2025}\natexlab{}.
\newblock \showarticletitle{A framework for designing provably secure
  steganography}. In \bibinfo{booktitle}{\emph{34th USENIX Security Symposium
  (USENIX Security 25)}}. \bibinfo{pages}{6837--6856}.
\newblock


\bibitem[Ni et~al\mbox{.}(2023)]%
        {wild}
\bibfield{author}{\bibinfo{person}{Jinjie Ni}, \bibinfo{person}{Fuzhao Xue},
  \bibinfo{person}{Yuntian Deng}, \bibinfo{person}{Jason Phang},
  \bibinfo{person}{Kabir Jain}, \bibinfo{person}{Mahir~Hitesh Shah},
  \bibinfo{person}{Zangwei Zheng}, {and} \bibinfo{person}{Yang You}.}
  \bibinfo{year}{2023}\natexlab{}.
\newblock \bibinfo{title}{Instruction in the Wild: A User-based Instruction
  Dataset}.
\newblock
  \bibinfo{howpublished}{\url{https://github.com/XueFuzhao/InstructionWild}}.
\newblock


\bibitem[Sudan(1997)]%
        {listdecoding2}
\bibfield{author}{\bibinfo{person}{Madhu Sudan}.}
  \bibinfo{year}{1997}\natexlab{}.
\newblock \showarticletitle{Decoding of Reed Solomon Codes beyond the
  Error-Correction Bound}.
\newblock \bibinfo{journal}{\emph{J. Complex.}} \bibinfo{volume}{13},
  \bibinfo{number}{1} (\bibinfo{date}{March} \bibinfo{year}{1997}),
  \bibinfo{pages}{180–193}.
\newblock
\showISSN{0885-064X}
\urldef\tempurl%
\url{https://doi.org/10.1006/jcom.1997.0439}
\showDOI{\tempurl}


\bibitem[Walker(1974)]%
        {alisa2}
\bibfield{author}{\bibinfo{person}{Alastair~J Walker}.}
  \bibinfo{year}{1974}\natexlab{}.
\newblock \showarticletitle{New fast method for generating discrete random
  numbers with arbitrary frequency distributions}.
\newblock \bibinfo{journal}{\emph{Electronics Letters}} \bibinfo{volume}{10},
  \bibinfo{number}{8} (\bibinfo{year}{1974}), \bibinfo{pages}{127--128}.
\newblock


\bibitem[Walker(1977)]%
        {alisa1}
\bibfield{author}{\bibinfo{person}{Alastair~J Walker}.}
  \bibinfo{year}{1977}\natexlab{}.
\newblock \showarticletitle{An efficient method for generating discrete random
  variables with general distributions}.
\newblock \bibinfo{journal}{\emph{ACM Transactions on Mathematical Software
  (TOMS)}} \bibinfo{volume}{3}, \bibinfo{number}{3} (\bibinfo{year}{1977}),
  \bibinfo{pages}{253--256}.
\newblock


\bibitem[Wang et~al\mbox{.}(2025)]%
        {wang2025sparsamp}
\bibfield{author}{\bibinfo{person}{Yaofei Wang}, \bibinfo{person}{Gang Pei},
  \bibinfo{person}{Kejiang Chen}, \bibinfo{person}{Jinyang Ding},
  \bibinfo{person}{Chao Pan}, \bibinfo{person}{Weilong Pang},
  \bibinfo{person}{Donghui Hu}, {and} \bibinfo{person}{Weiming Zhang}.}
  \bibinfo{year}{2025}\natexlab{}.
\newblock \showarticletitle{SparSamp: Efficient Provably Secure Steganography
  Based on Sparse Sampling}. In \bibinfo{booktitle}{\emph{34th USENIX Security
  Symposium (USENIX Security 25)}}.
\newblock


\bibitem[Yang et~al\mbox{.}(2020a)]%
        {TS-CSW}
\bibfield{author}{\bibinfo{person}{Zhongliang Yang}, \bibinfo{person}{Yongfeng
  Huang}, {and} \bibinfo{person}{Yujin Zhang}.}
  \bibinfo{year}{2020}\natexlab{a}.
\newblock \showarticletitle{TS-CSW: Text steganalysis and hidden capacity
  estimation based on convolutional sliding windows}.
\newblock \bibinfo{journal}{\emph{Multimedia Tools and Applications}}
  \bibinfo{volume}{79} (\bibinfo{year}{2020}), \bibinfo{pages}{18293--18316}.
\newblock


\bibitem[Yang et~al\mbox{.}(2018)]%
        {yang2018rnn}
\bibfield{author}{\bibinfo{person}{Zhong-Liang Yang},
  \bibinfo{person}{Xiao-Qing Guo}, \bibinfo{person}{Zi-Ming Chen},
  \bibinfo{person}{Yong-Feng Huang}, {and} \bibinfo{person}{Yu-Jin Zhang}.}
  \bibinfo{year}{2018}\natexlab{}.
\newblock \showarticletitle{RNN-stega: Linguistic steganography based on
  recurrent neural networks}.
\newblock \bibinfo{journal}{\emph{IEEE Transactions on Information Forensics
  and Security}} \bibinfo{volume}{14}, \bibinfo{number}{5}
  (\bibinfo{year}{2018}), \bibinfo{pages}{1280--1295}.
\newblock


\bibitem[Yang et~al\mbox{.}(2020b)]%
        {yang2020vae}
\bibfield{author}{\bibinfo{person}{Zhong-Liang Yang}, \bibinfo{person}{Si-Yu
  Zhang}, \bibinfo{person}{Yu-Ting Hu}, \bibinfo{person}{Zhi-Wen Hu}, {and}
  \bibinfo{person}{Yong-Feng Huang}.} \bibinfo{year}{2020}\natexlab{b}.
\newblock \showarticletitle{VAE-Stega: linguistic steganography based on
  variational auto-encoder}.
\newblock \bibinfo{journal}{\emph{IEEE Transactions on Information Forensics
  and Security}}  \bibinfo{volume}{16} (\bibinfo{year}{2020}),
  \bibinfo{pages}{880--895}.
\newblock


\bibitem[Ziegler et~al\mbox{.}(2019)]%
        {ziegler2019neural}
\bibfield{author}{\bibinfo{person}{Zachary Ziegler}, \bibinfo{person}{Yuntian
  Deng}, {and} \bibinfo{person}{Alexander~M Rush}.}
  \bibinfo{year}{2019}\natexlab{}.
\newblock \showarticletitle{Neural Linguistic Steganography}. In
  \bibinfo{booktitle}{\emph{Proceedings of the 2019 Conference on Empirical
  Methods in Natural Language Processing and the 9th International Joint
  Conference on Natural Language Processing (EMNLP-IJCNLP)}}.
  \bibinfo{pages}{1210--1215}.
\newblock


\end{thebibliography}
\end{document}